\documentclass[journal]{IEEEtran}

\usepackage{amssymb}
\usepackage{amsmath}
\usepackage{cite}
\usepackage{url}
\usepackage{empheq}
\usepackage{xcolor}
\usepackage{graphicx}
\usepackage{subfigure}
\usepackage{enumitem}
\usepackage{fancyhdr}
\usepackage{mdwmath}	
\usepackage{mdwtab}
\usepackage{caption}
\usepackage{amsthm}
\usepackage{algorithm}
\usepackage{algorithmic}

\newtheorem{lemma}{Lemma}
\newtheorem{remark}{Remark}
\newtheorem{theorem}{Theorem}
\newtheorem{corollary}{Corollary}
\newtheorem{assumption}{Assumption}
\newtheorem{definition}{Definition}
\newtheorem{proposition}{Proposition}

\newcommand{\eqr}[1]{(\ref{#1})}
\newcommand{\fref}[1]{Fig.~\ref{#1}}

\newcommand*{\QEDA}{\null\nobreak\hfill\ensuremath{\blacksquare}}

\begin{document}
\title{Age-of-Information Aware Federated Learning with Finite Speed Pinching Antenna}
\author{Kaidi~Wang,~\IEEEmembership{Member,~IEEE,}
Daniel~K.~C.~So,~\IEEEmembership{Senior~Member,~IEEE,}
and~Zhiguo~Ding,~\IEEEmembership{Fellow,~IEEE}
\thanks{Kaidi~Wang and Daniel~K.~C.~So are with the Department of Electrical and Electronic Engineering, the University of Manchester, Manchester, M1 9BB, UK (email: kaidi.wang@ieee.org; d.so@manchester.ac.uk).}
\thanks{Zhiguo~Ding is with the School of Electrical and Electronic Engineering (EEE), Nanyang Technological University, Singapore 639798 (e-mail: zhiguo.ding@ntu.edu.sg).}}
\maketitle
\begin{abstract}
This paper investigates age-of-information (AoI) aware federated learning over wireless networks with finite speed pinching antennas. In contrast to existing studies that assume an infinitely high antenna moving speed, a practical round based training procedure is considered, where the pinching antenna is repositioned during the local training phase and its feasible movement range depends on the selected devices. This creates a new coupling among device selection, antenna placement, local training time, model uploading time, and AoI evolution. To characterize the impact of antenna moving speed, the rate gain over the fixed antenna and the gap to the infinite speed benchmark are analyzed. Subsequently, an overall AoI minimization problem is formulated under a round latency deadline by jointly optimizing the selected device set and the pinching antenna position. A coalitional game based device selection algorithm is proposed, where finite speed antenna placement is incorporated into the coalition utility evaluation. For antenna placement, the optimal search region is derived by exploiting the mobility constraint and device location span, based on which a branch-and-bound (BnB) algorithm is developed to obtain the global optimum. Simulation results show that the proposed scheme can accelerate learning convergence, reduce the sum AoI, and improve device participation compared with baseline schemes, demonstrating the potential of pinching antennas for enhancing federated learning through flexible spatial reconfiguration.
\end{abstract}
\section{Introduction}
Federated learning has emerged as a promising paradigm for enabling privacy-preserving edge intelligence, where distributed devices collaboratively train a shared model without uploading their raw data to a central server \cite{mcmahan2017fl}. In wireless networks, however, the performance of federated learning is significantly affected by device heterogeneity and communication conditions \cite{yang2020ds, chen2020fl}. Due to diverse computation capabilities, local dataset sizes, and channel gains, devices may experience substantially different local training and model uploading delays, resulting in the straggler effect and prolonging the global model aggregation process \cite{gafni2022fl, niknam2020fl}. Moreover, repeatedly selecting devices with favorable computation or communication conditions may lead to unbalanced device participation, stale local updates, and biased global aggregation, especially under non-independent and identically distributed (non-IID) data distributions \cite{lei2023fl, kaidi2025fl2}. Accordingly, reducing communication latency while ensuring balanced and timely device participation remains a fundamental challenge in wireless federated learning.

In this context, pinching antennas provide a promising physical layer solution for mitigating communication induced straggler effects in wireless federated learning. By activating a dielectric particle on a waveguide, a pinching antenna creates a radiating point whose location can be flexibly adjusted along the waveguide \cite{ding2024pin}. Consequently, the wireless propagation environment can be partially reconfigured through antenna positioning, enabling the server to enhance the uplink channels of selected devices and reduce their model uploading delays \cite{yang2025pinching, liu2025pinching}. This capability fundamentally changes the conventional device selection logic in wireless federated learning. Instead of passively adapting device participation to existing channel conditions, the network can actively reshape the uplink geometry to support devices that are more valuable for learning, such as those with stale updates, diverse local data, or high training importance \cite{fang2025pinching, kaidi2025pin}. Therefore, pinching antennas introduce a new spatial design dimension for jointly improving communication efficiency and learning aware device participation in wireless federated learning \cite{xu2025generalized, kaidi2025generalized}.
\subsection{Related Works}
Motivated by these advantages, several recent studies have investigated the integration of pinching antennas into wireless federated learning \cite{lin2025pinching, wu2026straggler, asaad2026energy, asaad2026fedpass, lin2026tail}. In \cite{lin2025pinching}, a single-waveguide pinching-antenna system was studied for both synchronous and asynchronous federated learning, where the active radiating point is adjusted along the waveguide to improve the worst uplink channels. In the synchronous setting, antenna positioning is designed to reduce the round latency determined by the slowest selected device, whereas in the asynchronous setting, it is employed to increase the number of devices that complete model uploading within a given deadline. To further exploit the communication benefits of pinching antennas, \cite{wu2026straggler} proposed a hybrid conventional and pinching antenna network for non-orthogonal multiple access (NOMA) enabled federated learning, where devices are classified according to their data contributions and communication conditions, and the total training time is minimized by jointly optimizing pinching antenna placement and communication resource allocation via reinforcement learning. In \cite{asaad2026energy}, pinching antennas were incorporated into over-the-air federated learning, where a server with multiple pinching antennas on a single waveguide aggregates local model updates through analog over-the-air computation. The corresponding design jointly optimizes pinching antenna positions, device scheduling, and analog transmission scaling to minimize aggregation energy under a target accuracy requirement. From a latency-learning tradeoff perspective, \cite{asaad2026fedpass} developed a low latency digital federated learning framework with multiple pinching antennas and time division multiple access (TDMA), where device scheduling, communication time allocation, transmit power, local computing frequency, and pinching antenna positions are optimized to balance the end-to-end round latency and an upper bound on the learning optimality gap. In a related direction, \cite{lin2026tail} investigated a tail latency aware pinching-antenna assisted federated learning system, where a single pinching antenna reshapes the uplink latency distribution of sampled devices, and the pinching antenna position and device participation probabilities are jointly designed to minimize the expected time-to-accuracy by coupling the maximum order statistic round latency with a heterogeneity aware convergence factor.
\subsection{Motivation and Contributions}
Although these studies have demonstrated the benefits of pinching antennas for wireless federated learning, the antenna positioning process is generally assumed to be instantaneous, such that the optimized antenna position is obtained without explicitly considering the movement time along the waveguide. More broadly, most existing studies on pinching-antenna systems focus on the final antenna location after reconfiguration \cite{xu2025pin, wang2025pa, tegos2025pin, ouyang2025uplink, kaidi2025pin2}, while the impact of finite antenna moving speed remains largely unexplored. In practice, antenna repositioning is not instantaneous, and recent studies have considered finite antenna movement by incorporating the associated motion energy consumption \cite{liu2026energy, xu2026pin2}. Such antenna repositioning is particularly relevant to federated learning, where the pinching antenna is repositioned sequentially across communication rounds from its previous location. Moreover, since each communication round consists of a local training phase followed by a model uploading phase, the local training phase naturally provides a time window for repositioning the pinching antenna before uplink transmission.

In the round based procedure, the repositioning time is not fixed but depends on the selected devices, whose heterogeneous computational workloads and processing capabilities determine the local training duration and hence the feasible movement range of the pinching antenna. This leads to a distinctive tradeoff, where a device with a long local training time is not necessarily unfavorable, as it may provide a larger antenna movement budget despite increasing the training delay. Therefore, device selection and antenna mobility are mutually coupled in wireless federated learning. That is, the selected devices determine the feasible antenna movement range through the local training duration, while the resulting antenna position determines the uplink channel gains and uploading time. This mobility dependent coupling makes device selection fundamentally different from existing communication or learning oriented designs. To balance round latency and device participation fairness under this coupling, age-of-information (AoI) is adopted as the device selection criterion \cite{kaidi2024fl}, as it characterizes the freshness of device updates across communication rounds. The main contributions of this paper can be summarized as follows:
\begin{itemize}[leftmargin=*]
\item A pinching-antenna assisted federated learning framework with finite antenna moving speed is proposed, where the pinching antenna is repositioned during the local training phase and its feasible movement range is determined by the corresponding training duration. Based on the AoI evolution model that captures the effects of device selection and round latency, an overall AoI minimization problem is formulated over device selection and antenna placement under a round latency deadline to ensure timely global aggregation.
\item The impact of finite antenna moving speed on the uplink rate is analyzed under the training time constrained antenna movement model. By comparing the fixed antenna, speed limited pinching-antenna, and unlimited speed pinching-antenna cases, the rate gain and the rate gap are derived with explicit bounds. The resulting saturation moving speed reveals how the antenna moving speed and local training time jointly determine the effectiveness of antenna repositioning in reducing model uploading time.
\item The device selection problem is formulated as a coalitional game. A system level coalition utility is defined to capture both the freshness improvement brought by device participation and the round latency induced by local training and model uploading. A join-and-leave device selection algorithm is then developed to iteratively refine the selected device set, where the antenna placement subject to finite moving speed is incorporated into the utility evaluation to calculate the corresponding uploading time.
\item The antenna placement problem with finite antenna moving speed is studied for a given selected device set. For the single-device case, a closed-form optimal placement solution is derived. For the multi-device case, the search region is tightened by exploiting the feasible movement interval and the horizontal span of the selected devices. The optimal boundary solution is identified when these two intervals are disjoint, while a branch-and-bound (BnB) based algorithm is developed to obtain a globally optimal antenna placement over the reduced non-convex search region.
\end{itemize}
In addition, simulation results are presented to verify the effectiveness of the proposed framework under different learning tasks and system configurations. The results show that pinching antenna placement under finite moving speed can significantly improve learning convergence and AoI performance compared with benchmark schemes. Moreover, the impacts of antenna moving speed, transmit power, and latency deadline on device participation are characterized, further demonstrating the role of antenna repositioning in improving the freshness-latency tradeoff in wireless federated learning.
\section{System Model}
A pinching-antenna assisted federated learning network is considered, where the server is connected to a waveguide equipped with one pinching antenna to serve $N$ single-antenna devices. The set of all devices is denoted by $\mathcal{N}=\{1,2,\dots,N\}$. The service area is modeled as a rectangular region with length $D_x$ and width $D_y$, with its center set as the origin. The waveguide, with length $D_x$, is deployed along the $x$-axis at $y=0$ and at height $d$. The feed point is located at the left end of the waveguide and is given by $\boldsymbol{\psi}_\mathrm{feed}=(-D_x/2,0,d)$.
\subsection{Pinching-Antenna assisted Federated Learning Framework}
In the considered pinching-antenna assisted federated learning network, model training is performed over multiple communication rounds. At the beginning of round $t$, the server determines the selected device set $\mathcal{S}_t$ and broadcasts the current global model $\mathbf{w}^{(t)}$, together with the device selection result, to all devices\footnote{In this work, the optimization focuses on the local training and model uploading phases. The downlink transmission of the global model and the device selection notification is assumed to be completed before local training. During this phase, the pinching antenna remains at the position reached in the previous communication round.}. After receiving this information, each selected device $n\in\mathcal{S}_t$ performs local training on its local dataset and obtains the updated local model as
\begin{equation}
\mathbf{w}_n^{(t)}=\mathbf{w}^{(t)}-\alpha\mathbf{g}_n^{(t)},
\end{equation}
where $\alpha$ is the learning rate, and $\mathbf{g}_n^{(t)}$ is the local gradient or stochastic gradient computed at device $n$. During the local training phase, the pinching antenna is repositioned along the waveguide from $\boldsymbol{\psi}_\mathrm{pin}^{(t-1)}$ to $\boldsymbol{\psi}_\mathrm{pin}^{(t)}$, such that the local training duration serves as the time budget for antenna movement\footnote{In this work, the local training phase is selected as the antenna repositioning interval since it does not involve model transmission. In practice, antenna movement may also be performed concurrently with wireless transmission, as considered in antenna roaming \cite{kaidi2026ar}, which results in time-varying channels and requires joint antenna trajectory and transmission design.}. After local training, the selected devices sequentially upload their updated local models $\mathbf{w}_n^{(t)}$ to the server via TDMA. The server then aggregates the received local models according to FedAvg \cite{mcmahan2017fl}, as follows:
\begin{equation}
\mathbf{w}^{(t+1)}=\frac{\sum_{n\in\mathcal{S}_t}\beta_n\mathbf{w}_n^{(t)}}{\sum_{i\in\mathcal{S}_t}\beta_i},
\end{equation}
where $\beta_n$ is the number of local training samples at device $n$.
\subsection{Local Model Training and Pinching-Antenna Mobility}
The local training time depends on the computational workloads and central processing unit (CPU) capabilities of devices. For device $n$, the local training time is given by
\begin{equation}
T_{\mathrm{tr},n}=\frac{\mu\beta_n}{C_n},
\end{equation}
where $\mu$ is the required number of CPU cycles for training one sample, and $C_n$ is the CPU capability of device $n$. According to the considered round based procedure, the local training time in communication round $t$ is determined by the slowest selected device, i.e.,
\begin{equation}
T_\mathrm{tr}^{(t)}=\max\left\{T_{\mathrm{tr},n}| n\in\mathcal{S}_t\right\}.
\end{equation}

During local training, the pinching antenna is repositioned along the waveguide. Since antenna positioning is performed within the local training period\footnote{In practical implementation, antenna repositioning may involve additional control and signaling overhead. However, the physical repositioning delay has been recognized as the dominant factor and a major practical limitation in flexible antenna systems \cite{li2026movable, wang2026movable}. Therefore, this work explicitly considers the repositioning delay caused by the finite antenna moving speed, while the additional control and signaling overhead is neglected.}, the movement of the pinching antenna in round $t$ is constrained by
\begin{equation}
\left|x_\mathrm{pin}^{(t)}-x_\mathrm{pin}^{(t-1)}\right|\le v_\mathrm{pin}T_\mathrm{tr}^{(t)},
\end{equation}
where $v_\mathrm{pin}$ is the moving speed of the pinching antenna. Equivalently, the horizontal position of the pinching antenna satisfies
\begin{equation}\label{range}
x_\mathrm{pin}^{(t)}\in\left[x_\mathrm{pin}^{(t-1)}-v_\mathrm{pin}T_\mathrm{tr}^{(t)}, x_\mathrm{pin}^{(t-1)}+v_\mathrm{pin}T_\mathrm{tr}^{(t)}\right].
\end{equation}
In the considered system, the local training time serves as the available time budget for antenna repositioning. Given the previous antenna position and the finite moving speed, the feasible antenna placement range in each communication round is specified by \eqr{range}, and positions requiring a longer repositioning time are regarded as infeasible. Meanwhile, due to the finite waveguide length, the antenna position also satisfies $x_\mathrm{pin}^{(t)}\in [-D_x/2,D_x/2]$. Without loss of generality, the initial horizontal position of the pinching antenna is set at the center of the waveguide, i.e., $x_\mathrm{pin}^{(0)}=0$. Therefore, the conventional fixed-location antenna case can be viewed as a special case by setting $v_\mathrm{pin}=0$, where the antenna remains fixed over all communication rounds.

The above mobility model reveals a unique coupling between device selection and antenna positioning, as follows.
\begin{remark}
In the proposed model, the feasible movement range of the pinching antenna in communication round $t$ is not fixed in advance but is determined by the local training time $T_\mathrm{tr}^{(t)}$. Since $T_\mathrm{tr}^{(t)}$ depends on the selected device set $\mathcal{S}_t$, device selection affects not only the learning and communication performance but also the feasible antenna position region in the same round.
\end{remark}

\begin{remark}
Given a selected device set $\mathcal{S}_t$, if an additional device $n\notin\mathcal{S}_t$ has a local training time longer than those of all selected devices, the feasible antenna position region can be enlarged accordingly. Therefore, in the proposed model, device participation can also have a positive effect on communication by expanding the antenna movement budget and potentially improving the uplink communication performance.
\end{remark}
\subsection{Pinching-Antenna based Local Model Uploading}
In the pinching-antenna system, the effective channel from the server to each device consists of the in-waveguide propagation path and the free-space propagation path. For the considered channel model, perfect channel state information (CSI) of all devices is assumed to be available at the server. In communication round $t$, the channel from the server to device $n$ can be modeled as follows:
\begin{equation}\label{channel}
h_n^{(t)}=\frac{\eta e^{-j\left(\frac{2\pi}{\lambda}\left\|\boldsymbol{\psi}_n-\boldsymbol{\psi}_\mathrm{pin}^{(t)}\right\|+\frac{2\pi}{\lambda_g}\left\|\boldsymbol{\psi}_\mathrm{feed}-\boldsymbol{\psi}_\mathrm{pin}^{(t)}\right\|\right)}}{\left\|\boldsymbol{\psi}_n-\boldsymbol{\psi}_\mathrm{pin}^{(t)}\right\|},
\end{equation}
where $\eta=\frac{c}{4\pi f_c}$, $c$ is the speed of light, $f_c$ is the carrier frequency, $\lambda$ is the carrier wavelength, and $\lambda_g=\frac{\lambda}{n_\mathrm{eff}}$ is the wavelength in the dielectric waveguide, with $n_\mathrm{eff}$ denoting the effective refractive index. Moreover, $\boldsymbol{\psi}_n=(x_n,y_n,0)$ is the location of device $n$, $\boldsymbol{\psi}_\mathrm{pin}^{(t)}=(x_\mathrm{pin}^{(t)},0,d)$ is the location of the pinching antenna in communication round $t$, and $\|\cdot\|$ denotes the Euclidean norm. Note that the in-waveguide propagation attenuation is neglected in \eqr{channel}, as its impact can be insignificant under practical system conditions \cite{xu2026attenuation}.

After local training is completed at all selected devices, the local model uploading phase is initiated\footnote{In this work, the local training period serves as the antenna repositioning interval. Accordingly, model uploading starts after all selected devices complete local training, with all model updates transmitted from a common antenna position determined for the current round. This design utilizes the available training time for repositioning and facilitates the joint design of device selection and antenna placement.}. During this phase, TDMA is adopted, where the selected devices sequentially upload their local models to the server over orthogonal time slots for aggregation. Accordingly, the achievable data rate of device $n$ in communication round $t$ is given by
\begin{equation}
R_n^{(t)}=B\log_2\left(1+\frac{P_n|h_n^{(t)}|^2}{\sigma^2}\right),
\end{equation}
where $B$ is the bandwidth, $P_n$ is the transmit power at device $n$, and $\sigma^2$ is the noise power. For device $n$, the corresponding local model uploading time is
\begin{equation}
T_{\mathrm{up},n}^{(t)}=\frac{D_\mathrm{lm}}{R_n^{(t)}},
\end{equation}
where $D_\mathrm{lm}$ is the data size of the local model. Therefore, the total model uploading time in communication round $t$ is
\begin{equation}
T_\mathrm{up}^{(t)}=\sum_{n\in\mathcal{S}_t}T_{\mathrm{up},n}^{(t)}.
\end{equation}
\subsection{Rate Gain Analysis under Antenna Mobility}
To characterize the impact of pinching antenna moving speed on the uplink rate, a special case with a single selected device $n$ is considered. The pinching antenna starts from the origin and moves toward the optimal antenna position, i.e., the projection of device $n$ onto the waveguide, during local training. The fixed antenna case and the unlimited speed pinching antenna case are considered as two benchmarks, corresponding to $v_\mathrm{pin}=0$ and $v_\mathrm{pin}=\infty$, respectively. For notational simplicity, the communication round index $t$ is omitted in this subsection. In this case, the achievable data rate at antenna position $x_\mathrm{pin}$ can be rewritten as follows:
\begin{equation}
R_n(x_\mathrm{pin})=B\log_2\left(1+\frac{a_n}{(x_\mathrm{pin}-x_n)^2+b_n}\right),
\end{equation}
where $a_n\triangleq P_n\eta^2/\sigma^2$ and $b_n\triangleq y_n^2+d^2$.
The saturation speed is defined as
\begin{equation}
v_{\mathrm{sat},n}=\frac{|x_n|}{T_{\mathrm{tr},n}},
\end{equation}
which represents the minimum antenna moving speed required to reach the projection of device $n$ within its local training time. If $0\le v_\mathrm{pin}<v_{\mathrm{sat},n}$, the residual horizontal distance between the pinching antenna and the projection of device $n$ is given by
\begin{equation}
r_n(v_\mathrm{pin})=|x_n|-v_\mathrm{pin}T_{\mathrm{tr},n}.
\end{equation}
When $v_\mathrm{pin}\ge v_{\mathrm{sat},n}$, the pinching antenna can reach the projection of device $n$ during local training, and thus becomes equivalent to the unlimited speed benchmark.

Based on these definitions, the impact of antenna moving speed on the data rate can be characterized as follows.
\begin{proposition}\label{rategap}
When $0\le v_\mathrm{pin}<v_{\mathrm{sat},n}$, the rate gain from the fixed antenna to the speed limited pinching antenna is
\begin{equation}
\Delta R_{n,0\to v}=B\log_2\!\left(1+\frac{a_n\left(x_n^2-r_n^2(v_\mathrm{pin})\right)}{\left(r_n^2(v_\mathrm{pin})+b_n\right)\left(x_n^2+b_n+a_n\right)}\right),
\end{equation}
with the lower bound
\begin{equation}
\Delta R_{n,0\to v}\ge\frac{B}{\ln 2}\frac{a_n\left(x_n^2-r_n^2(v_\mathrm{pin})\right)}{\left(x_n^2+b_n\right)\left(r_n^2(v_\mathrm{pin})+b_n+a_n\right)}.
\end{equation}
The remaining rate gap from the speed limited pinching antenna to the unlimited speed benchmark is
\begin{equation}
\Delta R_{n,v\to\infty}=B\log_2\!\left(1+\frac{a_n r_n^2(v_\mathrm{pin})}{b_n\left(r_n^2(v_\mathrm{pin})+b_n+a_n\right)}\right),
\end{equation}
with the upper bound
\begin{equation}
\Delta R_{n,v\to\infty}\le\frac{B}{\ln 2}\frac{a_n r_n^2(v_\mathrm{pin})}{b_n\left(r_n^2(v_\mathrm{pin})+b_n+a_n\right)}.
\end{equation}
Moreover, the lower bound of $\Delta R_{n,0\to v}$ is non-decreasing in $v_\mathrm{pin}$, whereas the upper bound of $\Delta R_{n,v\to\infty}$ is non-increasing in $v_\mathrm{pin}$. For $v_\mathrm{pin}\ge v_{\mathrm{sat},n}$, $\Delta R_{n,v\to\infty}=0$.
\end{proposition}
\begin{IEEEproof}
Refer to Appendix~A.
\end{IEEEproof}

This result provides the following insight.
\begin{remark}
Proposition~\ref{rategap} reveals that the rate performance depends on both the antenna moving speed and the local training time. Increasing $v_\mathrm{pin}$ or $T_{\mathrm{tr},n}$ reduces the residual distance $r_n(v_\mathrm{pin})$, thereby increasing the gain over the fixed antenna benchmark and reducing the gap to the unlimited speed benchmark. Once $v_\mathrm{pin}\ge v_{\mathrm{sat},n}$, further increasing the moving speed brings no additional uplink rate improvement.
\end{remark}

\begin{remark}
Proposition~\ref{rategap} shows that the rate improvement enabled by antenna repositioning is also affected by $b_n=y_n^2+d^2$. Specifically, a larger device-to-waveguide distance or waveguide height increases the non-horizontal propagation distance, thereby reducing the relative distance improvement achieved by horizontal antenna movement.
\end{remark}

Since the local model uploading time is inversely proportional to the achievable data rate, the above rate gain enabled by antenna repositioning can be equivalently interpreted as a reduction in model uploading time, which is incorporated into the following system optimization.
\subsection{Round Latency and AoI Modeling}
In federated learning, the staleness of local data updates has a significant impact on the training performance, especially under non-IID data distributions. As indicated in \cite{kaidi2025fl}, AoI \cite{yates2020aoi} can effectively characterize the freshness of local data in federated learning and improve the learning performance by facilitating timely device participation. Moreover, frequent selection of devices with favorable communication conditions may lead to unbalanced device participation and further cause weight divergence \cite{kaidi2025fl2}. Therefore, in this work, AoI is adopted as the optimization criterion to balance device participation while maintaining data freshness.

In communication round $t$, the round latency is determined by the selected devices and is given by
\begin{equation}
T^{(t)}=T_\mathrm{tr}^{(t)}+T_\mathrm{up}^{(t)}=\max\left\{T_{\mathrm{tr},n}| n\in\mathcal{S}_t\right\}+\!\!\sum_{n\in\mathcal{S}_t}\!T_{\mathrm{up},n}^{(t)}.
\end{equation}
Following \cite{kaidi2024fl}, the AoI of device $n$ in communication round $t$ is defined as follows:
\begin{equation}\label{aoi}
A_n^{(t)}=
\begin{cases}
A_n^{(t-1)}+T^{(t)}, & \text{if } n\notin\mathcal{S}_t,\\
0, & \text{if } n\in\mathcal{S}_t.
\end{cases}
\end{equation}
The initial AoI of each device is set to the same value, i.e., $A_n^{(0)}=A_0$, $\forall n\in\mathcal{N}$, where $A_0\ge 0$. It is worth noting that the AoI of all devices is maintained at the server. Specifically, at the end of each communication round, the server updates the AoI based on the received local models. For a selected device, its AoI is reset to zero only after its local training and model uploading have been successfully completed. At this point, the server has obtained the latest local update from the device, and its information is therefore regarded as fresh. In contrast, the AoI of an unselected device increases by the duration of the current communication round since no new local update is received.

The above AoI expression reveals the coupling among device selection, round latency, and antenna mobility, as summarized in the following remarks.
\begin{remark}
The overall AoI in communication round $t$ consists of the accumulated freshness penalty of non-selected devices and the latency induced AoI increment. Therefore, device selection needs to balance the freshness improvement against the round latency increase caused by the selected device set.
\end{remark}

\begin{remark}
Under TDMA based uploading, selecting an additional device incurs extra uploading time. Nevertheless, this device may still reduce the overall AoI by resetting its AoI and, if it increases the local training time, by enlarging the feasible antenna movement region for subsequent uploading.
\end{remark}
\section{Problem Formulation}
Based on the above latency and AoI models, in communication round $t$, the selected device set $\mathcal{S}_t$ and the pinching antenna position $x_\mathrm{pin}^{(t)}$ are jointly considered for overall AoI minimization. Moreover, a round latency deadline is incorporated to prevent excessive training and uploading delays in each communication round, thereby ensuring timely global model aggregation. Based on the AoI evolution model in \eqr{aoi}, the overall AoI in round $t$ is given by
\begin{equation}\label{overallaoi}
A^{(t)}=\sum_{n\in\mathcal{N}}\!A_n^{(t)}=\!\!\!\!\sum_{n\in\mathcal{N}\setminus\mathcal{S}_t}\!\!\!\!\!A_n^{(t-1)}+\left(N-|\mathcal{S}_t|\right)\left(T_\mathrm{tr}^{(t)}+T_\mathrm{up}^{(t)}\right).
\end{equation}
Accordingly, the overall AoI minimization problem is formulated as follows:
\begin{subequations}
\begin{empheq}{align}
\min_{\mathcal{S}_t, x_\mathrm{pin}^{(t)}}\quad & A^{(t)}\\
\textrm{s.t.}\quad 
& x_\mathrm{pin}^{(t)} \in \left[-\frac{D_x}{2}, \frac{D_x}{2}\right],\\
& \left|x_\mathrm{pin}^{(t)}-x_\mathrm{pin}^{(t-1)}\right|
\le v_\mathrm{pin}T_\mathrm{tr}^{(t)},\\
& \mathcal{S}_t \subseteq \mathcal{N},\\
& |\mathcal{S}_t| \ge 1,\\
& T^{(t)} \le T_\mathrm{max}.
\end{empheq}
\label{problem}
\end{subequations}\vspace{-2mm}\\
In problem~\eqr{problem}, constraints (\ref{problem}b) and (\ref{problem}c) specify the feasible position of the pinching antenna. Constraints (\ref{problem}d) and (\ref{problem}e) define the feasible device selection set and ensure that at least one device is selected for model update. Constraint (\ref{problem}f) sets the round latency deadline, with $T_\mathrm{max}$ denoting the latency threshold. In this paper, a round based optimization framework is considered, with device selection and antenna placement optimized for each communication round. It is worth noting that although problem \eqr{problem} focuses on minimizing the overall AoI in the current communication round, the optimization still involves cross-round dependence. Specifically, the AoI state incorporates the historical participation of each device, while the feasible antenna position in the current round depends on the antenna position obtained in the previous round.

Problem~\eqr{problem} is a mixed-integer non-convex problem due to the coupling between device selection and antenna positioning, where the selected device set affects both the round latency and the feasible antenna movement range. Since joint optimization involves combinatorial device selection and continuous non-convex antenna placement, problem~\eqr{problem} is decomposed into device selection and antenna placement problems. For device selection, the selected device set is determined with its antenna placement evaluated for latency calculation, as follows:
\begin{subequations}
\begin{empheq}{align}
\min_{\mathcal{S}_t}\quad & A^{(t)}\\
\textrm{s.t.}\quad & \text{(\ref{problem}d)}, \text{(\ref{problem}e)}, \text{and}\,\text{(\ref{problem}f)}.\nonumber
\end{empheq}
\label{dsproblem}
\end{subequations}\vspace{-2mm}\\
For a given selected device set $\mathcal{S}_t$, the freshness term, the number of non-selected devices, and the local training time in \eqr{overallaoi} are fixed with respect to $x_\mathrm{pin}^{(t)}$. Hence, the model uploading time $T_\mathrm{up}^{(t)}$ is the only term in \eqr{overallaoi} that depends on the antenna position. Accordingly, minimizing the overall AoI with respect to $x_\mathrm{pin}^{(t)}$ is equivalent to minimizing $T_\mathrm{up}^{(t)}$, and the antenna placement problem can be formulated as
\begin{subequations}
\begin{empheq}{align}
\min_{x_\mathrm{pin}^{(t)}}\quad & T_\mathrm{up}^{(t)}\\
\textrm{s.t.}\quad & \text{(\ref{problem}b)}, \text{and}\,\text{(\ref{problem}c)}.\nonumber
\end{empheq}
\label{pinproblem}
\end{subequations}\vspace{-2mm}\\
It should be noted that the latency deadline is included in the device selection problem. For a given device selection result $\mathcal{S}_t$, the local training time and the available latency budget for model uploading are fixed. Since the antenna position only affects the model uploading time, minimizing the uploading time also minimizes the round latency. Therefore, the latency constraint can be removed from the antenna placement problem without changing its optimal solution.
\section{Coalitional Game based Device Selection}
In each communication round, the selected devices naturally form a coalition for local training and model uploading. Since the contribution of each device depends on the entire selected set through the AoI reduction, latency, and feasible antenna movement range, a coalitional game framework is adopted to characterize the set dependent utility.
\subsection{Coalitional Game Formulation}
For the device selection problem~\eqr{dsproblem}, a coalition formation game $(\mathcal{N},v)$ is constructed, where the devices in $\mathcal{N}$ are the players and the selected device set $\mathcal{S}_t$ is regarded as a coalition. Due to the set dependent contribution of each device, individual payoffs are difficult to define, and thus a coalition utility is adopted to directly evaluate each candidate coalition. Since the objective of the considered device selection problem is to minimize the overall AoI, the coalition utility is defined as the negative system cost, i.e.,
\begin{equation}\label{coalitionutility}
v(\mathcal{S}_t)=
\begin{cases}
\tilde{v}(\mathcal{S}_t), & \text{if } \mathcal{S}_t\neq\emptyset \text{ and } T^{(t)}(\mathcal{S}_t)\le T_\mathrm{max}, \\
-\infty, & \text{otherwise},
\end{cases}
\end{equation}
where
\begin{equation}
\tilde{v}(\mathcal{S}_t)=-\!\!\!\!\sum_{n\in\mathcal{N}\setminus\mathcal{S}_t}\!\!\!\!A_n^{(t-1)}-(N-|\mathcal{S}_t|)T^{(t)}(\mathcal{S}_t),
\end{equation}
and $T^{(t)}(\mathcal{S}_t)$ denotes the round latency associated with coalition $\mathcal{S}_t$. Therefore, maximizing the coalition utility is equivalent to minimizing the overall AoI. The value $-\infty$ is assigned to coalitions that violate constraint (\ref{problem}e) or (\ref{problem}f), corresponding to an empty coalition or a coalition exceeding the latency deadline, respectively. In this way, infeasible coalitions can be excluded from the coalition formation process.
\subsection{Preference Relation and Update Rules}
Given the coalition utility defined above, each device determines its preference according to whether joining or leaving the coalition improves the coalition utility. Specifically, a join or leave operation is performed only if it leads to a strictly larger utility. Based on this preference relation, the coalition is updated through the following rules.
\begin{definition}[Join Rule]
Given the current coalition $\mathcal{S}_t$, any device $n\notin\mathcal{S}_t$ joins the coalition if and only if
\begin{equation}
v(\mathcal{S}_t\cup\{n\})>v(\mathcal{S}_t).
\end{equation}
\end{definition}
\begin{definition}[Leave Rule]
Given the current coalition $\mathcal{S}_t$ with $|\mathcal{S}_t|>1$, any device $n\in\mathcal{S}_t$ leaves the coalition if and only if
\begin{equation}
v(\mathcal{S}_t\setminus\{n\})>v(\mathcal{S}_t).
\end{equation}
\end{definition}

Therefore, both the join and leave operations are performed only when they strictly improve the coalition utility. To characterize the coalition evolution, let $\mathcal{S}_t^{(i)}$ denote the coalition obtained after the $i$-th update in communication round $t$, with $\mathcal{S}_t^{(0)}$ being the initial coalition. By iteratively performing such coalition updates, the selected device set is refined until no device can further improve the coalition utility by joining or leaving the coalition. At this point, the resulting coalition can be described in terms of Nash stability \cite{han2012game}. Specifically, the non-selected devices are regarded as an auxiliary coalition $\mathcal{N}\setminus\mathcal{S}_t$, such that the coalition structure can be represented by the two-coalition partition $\{\mathcal{S}_t,\mathcal{N}\setminus\mathcal{S}_t\}$. Based on this induced coalition structure, Nash stability is defined as follows.
\begin{definition}\label{nash}
A coalition structure $\{\mathcal{S}_t,\mathcal{N}\setminus\mathcal{S}_t\}$ is Nash-stable if no device can improve the coalition utility by unilaterally moving between the two coalitions, i.e.,
\begin{equation}
v(\mathcal{S}_t\cup\{n\})\le v(\mathcal{S}_t),\quad \forall n\notin\mathcal{S}_t,
\end{equation}
and
\begin{equation}
v(\mathcal{S}_t\setminus\{n\})\le v(\mathcal{S}_t),\quad \forall n\in\mathcal{S}_t,\ |\mathcal{S}_t|>1.
\end{equation}
\end{definition}

\begin{algorithm}[t]
\caption{Coalitional Game based Device Selection}
\label{coalitionalg}
\begin{algorithmic}[1]
\STATE \textbf{Input:} Device set $\mathcal{N}$ and system parameters.
\STATE \textbf{Initialization:} Set $i=0$ and $\mathcal{S}_t^{(0)}=\emptyset$.
\REPEAT
\STATE Set $\mathcal{S}_t^{(i+1)}=\mathcal{S}_t^{(i)}$.
\STATE Set $\mathcal{L}_\mathrm{join}$ as $\mathcal{N}\setminus\mathcal{S}_t^{(i+1)}$ sorted by descending $A_n^{(t-1)}$.
\FOR{$n\in\mathcal{L}_\mathrm{join}$}
\IF{$v(\mathcal{S}_t^{(i+1)}\cup\{n\})\!>\!v(\mathcal{S}_t^{(i+1)})$}
\STATE Update $\mathcal{S}_t^{(i+1)}=\mathcal{S}_t^{(i+1)}\cup\{n\}$.
\ENDIF
\ENDFOR
\STATE Set $\mathcal{L}_\mathrm{leave}$ as $\mathcal{S}_t^{(i+1)}$ sorted by ascending $A_n^{(t-1)}$.
\FOR{$n\in\mathcal{L}_\mathrm{leave}$}
\IF{$|\mathcal{S}_t^{(i+1)}|\!>\!1$ \textbf{and} $v(\mathcal{S}_t^{(i+1)}\setminus\{n\})\!>\!v(\mathcal{S}_t^{(i+1)})$}
\STATE Update $\mathcal{S}_t^{(i+1)}=\mathcal{S}_t^{(i+1)}\setminus\{n\}$.
\ENDIF
\ENDFOR
\STATE Set $i=i+1$.
\UNTIL{$\mathcal{S}_t^{(i)}=\mathcal{S}_t^{(i-1)}$}
\STATE \textbf{Output:} $\mathcal{S}_t=\mathcal{S}_t^{(i)}$.
\end{algorithmic}
\end{algorithm}

The resulting coalition formation process is summarized in Algorithm~\ref{coalitionalg}, where the initial coalition is set as the empty set to avoid dependence on a specific initialization. During the algorithm, the coalition is constructed through join operations and refined through leave operations according to the coalition utility. To reduce the checking overhead, non-selected devices are examined in descending order of AoI, while selected devices are examined in ascending order of AoI, since devices with larger AoI are more likely to provide a larger freshness gain when joining, whereas devices with smaller AoI tend to cause a smaller freshness loss when leaving. For each candidate coalition, the corresponding antenna placement is updated to obtain the minimum uploading time used in the utility evaluation.
\subsection{Properties Analysis}
The main properties of Algorithm~\ref{coalitionalg} are analyzed in terms of convergence, stability, and computational complexity.

\subsubsection{Convergence}
The convergence of Algorithm~\ref{coalitionalg} follows from the monotonic improvement of the coalition utility and the finite number of feasible coalitions. Specifically, in each iteration, one accepted join or leave operation strictly increases the coalition utility $v(\mathcal{S}_t)$ according to the proposed update rules. Since the number of feasible coalitions satisfying $\mathcal{S}_t\subseteq\mathcal{N}$ and $|\mathcal{S}_t|\ge 1$ is finite, the coalition utility cannot be improved indefinitely. Therefore, Algorithm~\ref{coalitionalg} terminates after a finite number of iterations.

\subsubsection{Stability}
Let $\mathcal{S}_t^\star$ denote the final coalition returned by Algorithm~\ref{coalitionalg}. The algorithm terminates only when no join or leave operation can further improve the coalition utility. Hence, no non-selected device can improve the utility by joining $\mathcal{S}_t^\star$, and no selected device can improve the utility by leaving $\mathcal{S}_t^\star$. According to Definition~\ref{nash}, the final coalition structure $\{\mathcal{S}_t^\star,\mathcal{N}\setminus\mathcal{S}_t^\star\}$ is Nash-stable.

\subsubsection{Complexity}
The computational complexity of Algorithm~\ref{coalitionalg} depends on the number of outer checking loops and the number of candidate devices evaluated in each loop. Let $C$ denote the number of outer loops before convergence. In each outer loop, at most $N$ candidate devices are evaluated for possible joining or leaving, and each accepted operation corresponds to one coalition update. Therefore, the overall computational complexity of Algorithm~\ref{coalitionalg} is $\mathcal{O}(CN)$.
\section{Branch-and-Bound based Antenna Placement}
In this section, the antenna placement problem~\eqr{pinproblem} is solved for a given device selection result $\mathcal{S}_t$. Since the waveguide propagation path only contributes a phase rotation in the adopted channel model, $|h_n^{(t)}|^2$ depends on the antenna position only through the free-space propagation distance between the pinching antenna and device $n$. By defining $a_n\triangleq P_n\eta^2/\sigma^2$ and $b_n\triangleq y_n^2+d^2$, the normalized uploading time of selected device $n$ can be written as follows:
\begin{equation}\label{noruptime}
f_n(x_\mathrm{pin}^{(t)})=\frac{1}{\log_2\left(1+\frac{a_n}{(x_\mathrm{pin}^{(t)}-x_n)^2+b_n}\right)}.
\end{equation}
By omitting the positive constant factor $D_\mathrm{lm}/B$, the antenna placement problem in \eqr{pinproblem} is equivalently expressed as
\begin{subequations}
\begin{empheq}{align}
\min_{x_\mathrm{pin}^{(t)}}\quad 
& \sum_{n\in\mathcal{S}_t}f_n(x_\mathrm{pin}^{(t)})\\
\textrm{s.t.}\quad & \text{(\ref{problem}b)}, \text{and}\,\text{(\ref{problem}c)}.\nonumber
\end{empheq}
\label{pinproblem1}
\end{subequations}\vspace{-2mm}\\
The solution to problem~\eqr{pinproblem1} is presented in the following subsections for the singleton coalition case and the general multi-device case, respectively.
\subsection{Optimal Placement for Singleton Device Case}
Before solving the multi-device case, the singleton device case is first examined, for which the optimal antenna placement admits a closed-form solution. When $\mathcal{S}_t=\{n\}$, the antenna placement problem reduces to minimizing a single term over the mobility constrained feasible interval, i.e.,
\begin{equation}
\min_{x_\mathrm{pin}^{(t)}\in[l^{(t)},u^{(t)}]} f_n(x_\mathrm{pin}^{(t)}),
\end{equation}
where
\begin{equation}
l^{(t)}=\max\left\{-\frac{D_x}{2}, x_\mathrm{pin}^{(t-1)}-v_\mathrm{pin}T_\mathrm{tr}^{(t)}\right\},
\end{equation}
and
\begin{equation}
u^{(t)}=\min\left\{\frac{D_x}{2}, x_\mathrm{pin}^{(t-1)}+v_\mathrm{pin}T_\mathrm{tr}^{(t)}\right\}.
\end{equation}
Since $f_n(x_\mathrm{pin}^{(t)})$ is monotonically increasing with respect to $(x_\mathrm{pin}^{(t)}-x_n)^2$, the optimal antenna position is the feasible point closest to $x_n$. Equivalently, it is obtained by projecting $x_n$ onto the feasible interval $[l^{(t)},u^{(t)}]$, i.e.,
\begin{equation}
x_\mathrm{pin}^{(t)\star}=
\begin{cases}
l^{(t)}, & \text{if } x_n<l^{(t)},\\
x_n, & \text{if } l^{(t)}\le x_n\le u^{(t)},\\
u^{(t)}, & \text{if } x_n>u^{(t)}.
\end{cases}
\end{equation}
This result indicates that the antenna is placed at the projection point of device $n$ if it is feasible; otherwise, it is placed at the nearest boundary of the feasible interval. Thus, the optimal antenna position in the singleton device case can be obtained without further optimization. In the following, the general case with $|\mathcal{S}_t|\ge 2$ is considered.
\subsection{Optimal Antenna Placement Region}
For the multi-device case, the antenna placement problem does not admit a direct closed-form solution due to the summation of multiple nonlinear uploading time terms. Nevertheless, the search region can be tightened by exploiting the monotonicity of each term with respect to the horizontal distance between the antenna and the selected devices. Recall that $[l^{(t)},u^{(t)}]$ denotes the mobility constrained feasible interval. Define the horizontal span of the selected devices as
\begin{equation}
x_\mathrm{min}^{(t)}=\min\{x_n| n\in\mathcal{S}_t\},
\end{equation}
and
\begin{equation}
x_\mathrm{max}^{(t)}=\max\{x_n| n\in\mathcal{S}_t\}.
\end{equation}
Here, $x_n$ is the signed horizontal coordinate of device $n$. The following proposition shows that, when the feasible interval intersects with the span of the selected devices, the global optimum lies in their intersection.
\begin{proposition}\label{feasibleregion}
Define $\underline{x}^{(t)}=\max\{l^{(t)},x_\mathrm{min}^{(t)}\}$ and $\overline{x}^{(t)}=\min\{u^{(t)},x_\mathrm{max}^{(t)}\}$. If $\underline{x}^{(t)}\le \overline{x}^{(t)}$, the globally optimal antenna position satisfies
\begin{equation}
x_\mathrm{pin}^{(t)\star}\in[\underline{x}^{(t)}, \overline{x}^{(t)}].
\end{equation}
\end{proposition}
\begin{IEEEproof}
Refer to Appendix~B.
\end{IEEEproof}

As a direct consequence of Proposition~\ref{feasibleregion}, when the mobility constrained feasible interval does not intersect with the span of the selected devices, the optimal antenna position can be obtained directly from the boundary of the feasible interval, as stated below.
\begin{corollary}\label{emptyregion}
If $\underline{x}^{(t)}>\overline{x}^{(t)}$, the optimal antenna position is given by
\begin{equation}
x_\mathrm{pin}^{(t)\star}=
\begin{cases}
u^{(t)}, & \text{if } u^{(t)}<x_\mathrm{min}^{(t)},\\
l^{(t)}, & \text{if } l^{(t)}>x_\mathrm{max}^{(t)}.
\end{cases}
\end{equation}
\end{corollary}
\begin{IEEEproof}
When $\underline{x}^{(t)}>\overline{x}^{(t)}$, the mobility constrained feasible interval $[l^{(t)},u^{(t)}]$ and the span of the selected devices $[x_\mathrm{min}^{(t)},x_\mathrm{max}^{(t)}]$ do not intersect. Therefore, either the entire feasible interval lies to the left of the selected device span, i.e., $u^{(t)}<x_\mathrm{min}^{(t)}$, or it lies to the right, i.e., $l^{(t)}>x_\mathrm{max}^{(t)}$.

If $u^{(t)}<x_\mathrm{min}^{(t)}$, increasing $x_\mathrm{pin}^{(t)}$ within $[l^{(t)},u^{(t)}]$ reduces the horizontal distance to all selected devices. Since each $f_n(x_\mathrm{pin}^{(t)})$ is monotonically increasing with respect to $(x_\mathrm{pin}^{(t)}-x_n)^2$, the objective function is monotonically decreasing, and the optimum is attained at $u^{(t)}$. Conversely, if $l^{(t)}>x_\mathrm{max}^{(t)}$, increasing $x_\mathrm{pin}^{(t)}$ enlarges the horizontal distance to all selected devices. Hence, the objective function is monotonically increasing, and the optimum is attained at $l^{(t)}$. The proof is completed.
\end{IEEEproof}

According to Proposition~\ref{feasibleregion} and Corollary~\ref{emptyregion}, the antenna placement problem in \eqr{pinproblem1} can be divided into two cases. If $\underline{x}^{(t)}>\overline{x}^{(t)}$, the optimal antenna position is directly obtained from the boundary of the mobility constrained feasible interval. Otherwise, when $\underline{x}^{(t)}\le\overline{x}^{(t)}$, problem~\eqr{pinproblem1} is reduced to the following problem:
\begin{equation}\label{pinproblem2}
\min_{x_\mathrm{pin}^{(t)}\in[\underline{x}^{(t)},\overline{x}^{(t)}]}\quad\sum_{n\in\mathcal{S}_t}f_n(x_\mathrm{pin}^{(t)}).
\end{equation}
This problem is a one-dimensional optimization problem with only the horizontal antenna position as the optimization variable. Moreover, based on Proposition~\ref{feasibleregion}, the feasible search region is further confined to a narrower interval $[\underline{x}^{(t)},\overline{x}^{(t)}]$. This structure makes BnB suitable for obtaining the globally optimal solution at a reduced computational cost. Therefore, the BnB method is adopted in this work to directly solve the antenna placement problem. Specifically, the BnB procedure iteratively partitions the interval, constructs valid lower and upper bounds for each subinterval, and prunes subintervals that cannot contain the global optimum \cite{boyd2007branch}.
\subsection{Interval based Bound Construction}
To implement the BnB procedure, valid lower and upper bounds need to be constructed over each generated interval. Consider an arbitrary interval $\mathcal{I}=[x_\mathrm{L},x_\mathrm{U}]\subseteq[\underline{x}^{(t)},\overline{x}^{(t)}]$. For each selected device $n\in\mathcal{S}_t$, define $d_n(\mathcal{I})$ as the minimum value of $(x_\mathrm{pin}^{(t)}-x_n)^2$ over $\mathcal{I}$, i.e.,
\begin{equation}
d_n(\mathcal{I})=
\begin{cases}
0, & \text{if } x_n\in[x_\mathrm{L},x_\mathrm{U}],\\
(x_\mathrm{L}-x_n)^2, & \text{if } x_n<x_\mathrm{L},\\
(x_n-x_\mathrm{U})^2, &  \text{if }x_n>x_\mathrm{U}.
\end{cases}
\end{equation}
That is, $d_n(\mathcal{I})=0$ when the projection point $x_n$ lies inside the interval $\mathcal{I}$; otherwise, $d_n(\mathcal{I})$ is the squared distance from $x_n$ to the nearest endpoint of $\mathcal{I}$. Therefore, for any $x_\mathrm{pin}^{(t)}\in\mathcal{I}$, the following inequality can be obtained:
\begin{equation}
(x_\mathrm{pin}^{(t)}-x_n)^2 \ge d_n(\mathcal{I}).
\end{equation}
Since $f_n(x_\mathrm{pin}^{(t)})$ is monotonically increasing with respect to $(x_\mathrm{pin}^{(t)}-x_n)^2$, it follows that
\begin{equation}
f_n(x_\mathrm{pin}^{(t)})\ge\frac{1}{\log_2\left(1+\frac{a_n}{d_n(\mathcal{I})+b_n}\right)}.
\end{equation}
By summing over all selected devices, a valid lower bound of the objective value over $\mathcal{I}$ is obtained as
\begin{equation}\label{lowerbound}
\ell(\mathcal{I})=\sum_{n\in\mathcal{S}_t}\frac{1}{\log_2\left(1+\frac{a_n}{d_n(\mathcal{I})+b_n}\right)}.
\end{equation}

For upper bounding, the objective value evaluated at any feasible point in $\mathcal{I}$ provides a valid upper bound. To obtain a tighter bound with low computational cost, the two endpoints and the midpoint of $\mathcal{I}$ are evaluated. With the midpoint
\begin{equation}
x_\mathrm{M}=\frac{x_\mathrm{L}+x_\mathrm{U}}{2},
\end{equation}
the upper bound over $\mathcal{I}$ is defined as
\begin{equation}\label{upperbound}
u(\mathcal{I})=\min\left\{\sum_{n\in\mathcal{S}_t}\!f_n(x_\mathrm{L}),\sum_{n\in\mathcal{S}_t}\!f_n(x_\mathrm{M}),\sum_{n\in\mathcal{S}_t}\!f_n(x_\mathrm{U})\right\}.
\end{equation}
Let $x(\mathcal{I})\in\{x_\mathrm{L},x_\mathrm{M},x_\mathrm{U}\}$ denote the feasible point attaining the minimum value in \eqr{upperbound}. Since $x(\mathcal{I})$ is a feasible point in $\mathcal{I}$, $u(\mathcal{I})$ is a valid upper bound.
\subsection{Branching, Bounding, and Pruning}
Let $\mathcal{P}$ denote the set of unpruned intervals. At each iteration, the interval with the smallest lower bound is selected for branching, i.e.,
\begin{equation}\label{activeinterval}
\mathcal{I}^\star=\arg\min_{\mathcal{I}\in\mathcal{P}}\ell(\mathcal{I}),
\end{equation}
where $\mathcal{I}^\star=[x_\mathrm{L}^\star,x_U^\star]$. Since problem~\eqr{pinproblem2} is one-dimensional, $\mathcal{I}^\star$ is bisected at its midpoint
\begin{equation}\label{midpoint}
x_\mathrm{M}^\star=\frac{x_\mathrm{L}^\star+x_U^\star}{2},
\end{equation}
which results in two subintervals $\mathcal{I}^-=[x_\mathrm{L}^\star,x_\mathrm{M}^\star]$ and $\mathcal{I}^+=[x_\mathrm{M}^\star,x_U^\star]$.

For the two subintervals $\mathcal{I}^-$ and $\mathcal{I}^+$, the corresponding lower and upper bounds are evaluated according to \eqr{lowerbound} and \eqr{upperbound}, respectively. Let $x_\mathrm{inc}$ denote the incumbent feasible point associated with the current global upper bound $U$. Define
\begin{equation}\label{ustar}
U^\star=\min\left\{U, u(\mathcal{I}^-),u(\mathcal{I}^+)\right\}.
\end{equation}
The incumbent feasible point is updated as
\begin{equation}\label{incupdate}
x_\mathrm{inc}=
\begin{cases}
x(\mathcal{I}^-), & \text{if } U^\star=u(\mathcal{I}^-),\\
x(\mathcal{I}^+), & \text{if } U^\star=u(\mathcal{I}^+),\\
x_\mathrm{inc}, & \text{otherwise}.
\end{cases}
\end{equation}
The global upper and lower bounds are then updated as
\begin{equation}\label{globalub}
U=U^\star,
\end{equation}
and
\begin{equation}\label{globallb}
L=\min_{\mathcal{I}\in\mathcal{P}}\ell(\mathcal{I}).
\end{equation}

Moreover, any unpruned interval $\mathcal{I}\in\mathcal{P}$ satisfying
\begin{equation}
\ell(\mathcal{I})\ge U-\varepsilon
\end{equation}
can be safely pruned, where $\varepsilon>0$ denotes the optimality tolerance. This pruning rule preserves global optimality, since no point in such an interval can achieve an objective value smaller than the current global upper bound by more than $\varepsilon$. Based on these procedures, the BnB based antenna placement algorithm is presented in Algorithm~\ref{bnbalg}.

\begin{algorithm}[t]
\caption{BnB based Antenna Placement for Problem~\eqr{pinproblem2}}
\label{bnbalg}
\begin{algorithmic}[1]
\STATE \textbf{Input:} Feasible interval $[\underline{x}^{(t)},\overline{x}^{(t)}]$, tolerance $\varepsilon>0$.
\STATE \textbf{Initialization:}
\STATE Set $\mathcal{I}^{(0)}=[\underline{x}^{(t)},\overline{x}^{(t)}]$ and $\mathcal{P}=\{\mathcal{I}^{(0)}\}$.
\STATE Compute $\ell(\mathcal{I}^{(0)})$ by \eqr{lowerbound}.
\STATE Compute $u(\mathcal{I}^{(0)})$ and $x(\mathcal{I}^{(0)})$ by \eqr{upperbound}.
\STATE Initialize $U=u(\mathcal{I}^{(0)})$, $L=\ell(\mathcal{I}^{(0)})$, and $x_\mathrm{inc}=x(\mathcal{I}^{(0)})$.
\WHILE{$U-L>\varepsilon$}
\STATE Select the interval $\mathcal{I}^\star$ by \eqr{activeinterval}, and remove $\mathcal{I}^\star$ from $\mathcal P$.
\STATE Compute the midpoint $x_\mathrm{M}^\star$ by \eqr{midpoint}.
\STATE Bisect $\mathcal{I}^\star$ into $\mathcal{I}^-$ and $\mathcal{I}^+$.
\STATE Compute $\ell(\mathcal{I}^-)$ and $\ell(\mathcal{I}^+)$ by \eqr{lowerbound}.
\STATE Compute $u(\mathcal{I}^-)$, $x(\mathcal{I}^-)$, $u(\mathcal{I}^+)$, and $x(\mathcal{I}^+)$ by \eqr{upperbound}.
\STATE Compute $U^\star$ by \eqr{ustar}.
\STATE Update the incumbent feasible point $x_\mathrm{inc}$ by \eqr{incupdate}.
\STATE Update the global upper bound $U$ by \eqr{globalub}.
\FOR{each $\mathcal{I}\in\{\mathcal{I}^-,\mathcal{I}^+\}$}
\IF{$\ell(\mathcal{I})<U-\varepsilon$}
\STATE Add $\mathcal{I}$ to $\mathcal{P}$.
\ENDIF
\ENDFOR
\STATE Update the global lower bound $L$ by \eqr{globallb}.
\ENDWHILE
\STATE \textbf{Output:} $x_\mathrm{inc}$ and $U$.
\end{algorithmic}
\end{algorithm}
\subsection{Properties Analysis}
The convergence, optimality, and complexity of Algorithm~\ref{bnbalg} are analyzed below.

\subsubsection{Convergence}
The proposed algorithm follows the standard BnB framework for one-dimensional global optimization. The lower bound $\ell(\mathcal{I})$ in \eqr{lowerbound} and the upper bound $u(\mathcal{I})$ in \eqr{upperbound} bound the local optimal objective value over each unpruned interval $\mathcal{I}\in\mathcal{P}$. During the iterations, $U$ is non-increasing, while $L$ is updated as the minimum lower bound among all unpruned intervals. Moreover, intervals satisfying $\ell(\mathcal{I})\ge U-\varepsilon$ are safely pruned, and the incumbent feasible point $x_\mathrm{inc}$ is updated together with $U$. Therefore, the proposed BnB procedure progressively reduces the search space while maintaining valid global lower and upper bounds.

\subsubsection{Complexity}
The computational complexity of the proposed BnB algorithm mainly depends on the number of intervals evaluated before termination. Let $N_\mathrm{int}$ denote the total number of generated intervals. For each interval $\mathcal{I}$, computing the lower bound in \eqr{lowerbound} requires evaluating $d_n(\mathcal{I})$ for all $n\in\mathcal{S}_t$, resulting in complexity $\mathcal{O}(|\mathcal{S}_t|)$. Similarly, computing the upper bound in \eqr{upperbound} requires evaluating the objective function at three feasible points, which also incurs complexity $\mathcal{O}(|\mathcal{S}_t|)$. Therefore, ignoring constant factors, the overall computational complexity is $\mathcal{O}(N_\mathrm{int}|\mathcal{S}_t|)$.

\subsubsection{Optimality}
By construction, $L$ and $U$ are valid global lower and upper bounds of problem~\eqr{pinproblem2}, respectively. Let $f^\star$ denote the globally optimal objective value. Then,
\begin{equation}
L\le f^\star\le U.
\end{equation}
Algorithm~\ref{bnbalg} terminates when $U-L\le\varepsilon$, which implies
\begin{equation}
U-f^\star\le U-L\le\varepsilon.
\end{equation}
Since $U$ is attained by the incumbent feasible point $x_\mathrm{inc}$, the output solution satisfies
\begin{equation}
\sum_{n\in\mathcal{S}_t}f_n(x_\mathrm{inc})=U\le f^\star+\varepsilon.
\end{equation}
Therefore, Algorithm~\ref{bnbalg} returns an $\varepsilon$-optimal solution to problem~\eqr{pinproblem2}.
\section{Simulation Results}
In this section, the learning and system performance of the proposed AoI aware device selection and pinching antenna placement scheme under finite moving speed are evaluated. For comparison, several benchmark schemes are considered. For antenna placement, the midpoint based scheme places the pinching antenna at the midpoint between the minimum and maximum $x$-coordinates of the selected devices. For device selection, the sequential selection scheme selects devices according to a fixed order from device $1$ to device $N$, where as many devices as possible are selected in each communication round until the latency threshold is violated. The random selection scheme generates a random ordering of devices in each round and sequentially adds devices according to this order until the latency threshold cannot be satisfied. The main simulation parameters are summarized in Table~\ref{parameter}.

\begin{table}[t]
\centering
\caption{Simulation Parameters}\vspace{-1mm}
\label{parameter}
\begin{tabular}{lc}
\hline
\textbf{Parameter} & \textbf{Value} \\ \hline
Number of devices ($N$) & $20$ \\
CPU cycles per sample ($\mu$) & $10^7$ \\
Carrier frequency ($f_c$) & $28$~GHz\\
Noise power ($\sigma^2$) & $-90$~dBm\\ 
Bandwidth ($B$) & $10$~MHz\\
Effective refractive index ($n_\mathrm{eff}$) & $1.4$ \\
Waveguide height ($d$) & $3$~m \\
Service area ($D_x \times D_y$) & $50$~m $\times$ $20$~m \\
Optimality tolerance $(\varepsilon)$ & $10^{-4}$ \\ \hline
\end{tabular}
\end{table}

\subsection{Learning Performance}
To evaluate the learning performance, MNIST and CIFAR-10 are adopted. For MNIST, $9000$ non-IID training samples are distributed among $N=20$ devices, where the local data size increases from $260$ to $640$ with a step size of $20$. Each device contains samples from only one label. For CIFAR-10, all $50000$ training samples are used, and the local data size increases from $2025$ to $2975$ with a step size of $50$. To construct a mild non-IID distribution for CIFAR-10, $50\%$ of each device's local data is drawn from an IID data pool, while the remaining $50\%$ is drawn from a single dominant label. Full-batch local training is adopted for both datasets, and each selected device performs one local epoch in each communication round. The stochastic gradient descent optimizer with a learning rate of $0.1$ is used for both datasets. For MNIST, a single hidden layer multilayer perceptron (MLP) with $256$ hidden neurons and ReLU activation is used, followed by a softmax output layer. For CIFAR-10, a lightweight convolutional neural network (CNN) is adopted, which consists of two $3\times3$ convolutional layers with $32$ filters, one max pooling layer, one $3\times3$ convolutional layer with $64$ filters, another max pooling layer, a fully connected layer with $128$ neurons, and a softmax output layer. For each dataset, the simulations are conducted with $10$ different initial models, and the presented learning curves are obtained by averaging the corresponding results. For better visualization, the learning curves are further smoothed using a moving average window of $20$ communication rounds.

\begin{figure}[!t]
\centering{
\subfigure[Antenna Position]{\centering{\includegraphics[width=82mm]{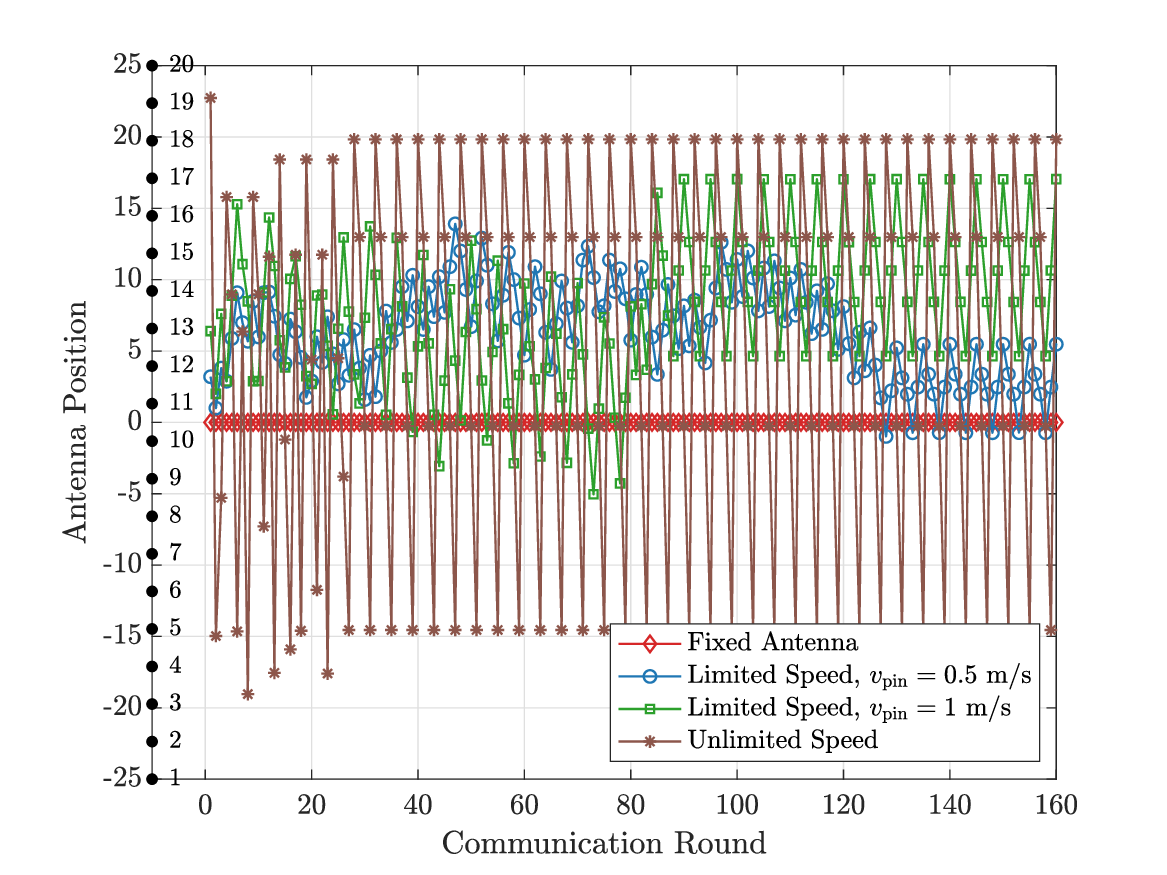}}}\vspace{-2mm}
\subfigure[Test Accuracy]{\centering{\includegraphics[width=82mm]{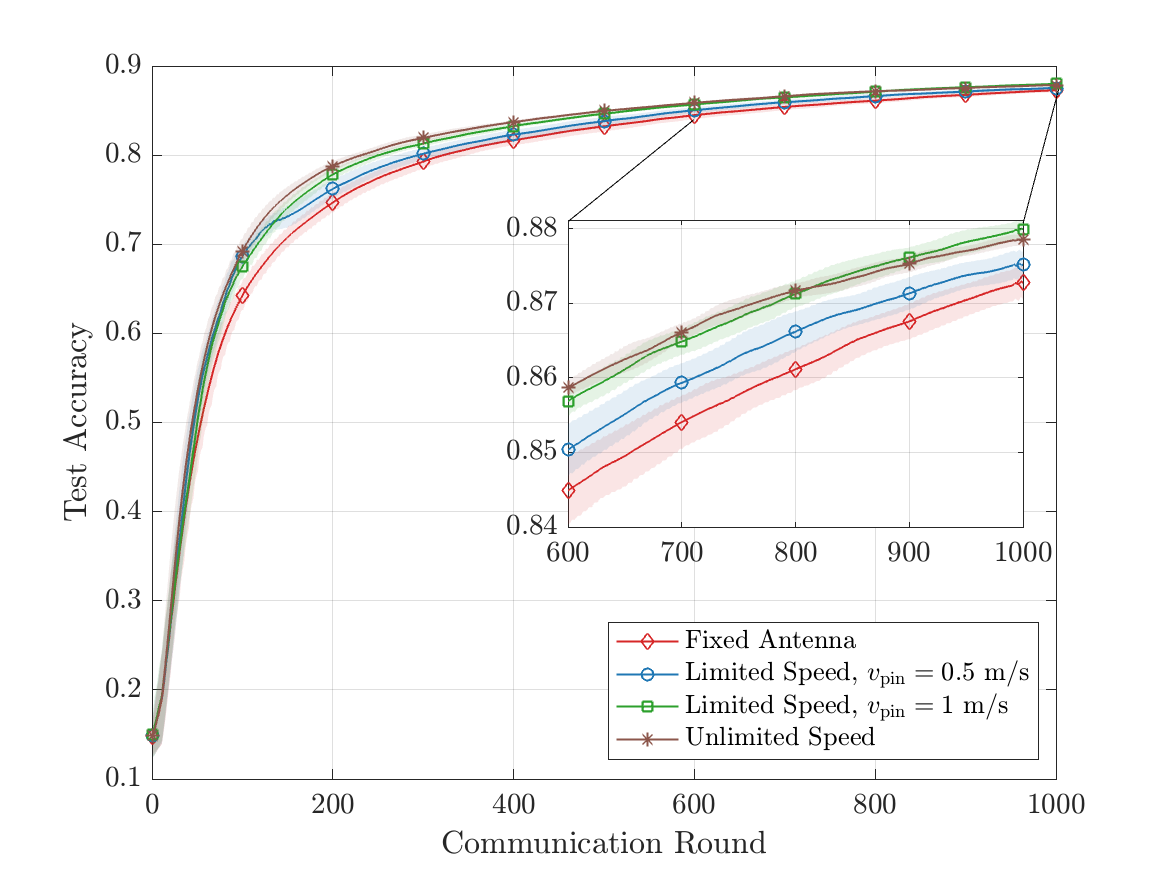}}}\vspace{-2mm}
\subfigure[Sum AoI]{\centering{\includegraphics[width=82mm]{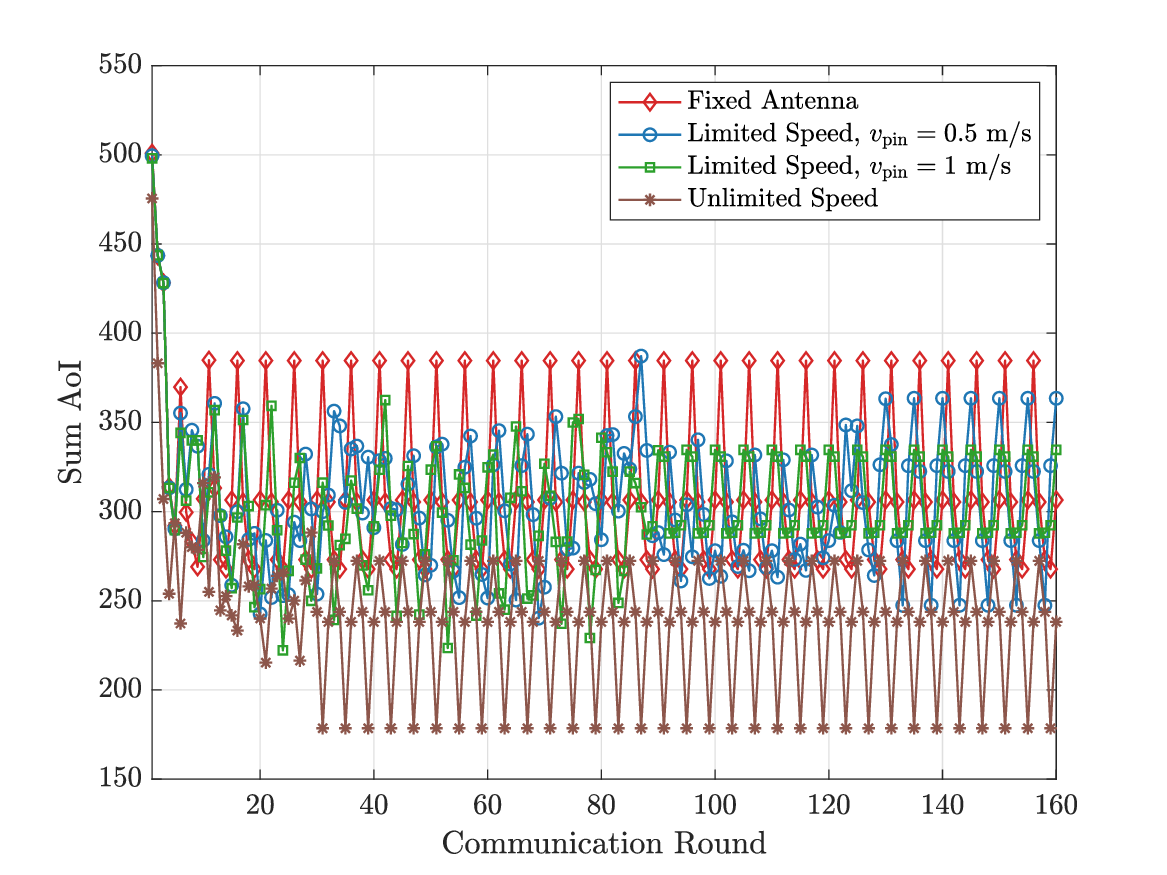}}}}\vspace{-2mm}
\caption{Learning performance on MNIST under different antenna moving speeds with fixed device deployment, where $C_n=1$~GHz, $D_\mathrm{lm}=50$~Mbits, $P_n=20$~dBm, $T_\mathrm{max}=8$~s, and $A_0=20$.}\vspace{-4mm}
\label{result1}
\end{figure}

\fref{result1} illustrates the impact of antenna moving speed on the learning performance under a unique device distribution, where the devices are uniformly deployed below the waveguide with an interval of $2.5$~m, as shown on the $y$-axis of \fref{result1}(a). It can be observed from \fref{result1}(b) that the learning performance improves as the antenna moving speed increases. In particular, even with a relatively low moving speed of $v_\mathrm{pin}=0.5$~m/s, the pinching antenna outperforms the fixed antenna case, accelerating convergence by about $100$ communication rounds. This observation indicates that clear learning performance gains can be achieved even under limited antenna mobility. Moreover, increasing the moving speed to $v_\mathrm{pin}=1$~m/s provides a further improvement of approximately $100$ rounds. Notably, the performance gap between $v_\mathrm{pin}=1$~m/s and the unlimited speed case is relatively small, indicating that a moderate moving speed is sufficient to approach the ideal unlimited speed case. These results demonstrate the applicability of the proposed mobility model to systems with different antenna repositioning capabilities. The antenna trajectories in \fref{result1}(a) show that the pinching antenna tends to move closer to devices with larger indices, especially under limited antenna moving speeds. This is because these devices have more local samples, which not only leave a smaller uploading time budget under the fixed deadline $T_\mathrm{max}$ due to longer training times, but also make their updates more important for learning performance. Meanwhile, a larger antenna moving speed enables the trajectory to reach a stable movement pattern more rapidly. In particular, the unlimited speed case becomes regular within about $30$ communication rounds, while the high-speed and low-speed cases require approximately $90$ and $130$ rounds, respectively. This is consistent with \fref{result1}(c), where the sum AoI also approaches a stable periodic behavior as the antenna movement becomes regular.

\begin{figure}[!t]
\centering{
\subfigure[Test Accuracy]{\centering{\includegraphics[width=82mm]{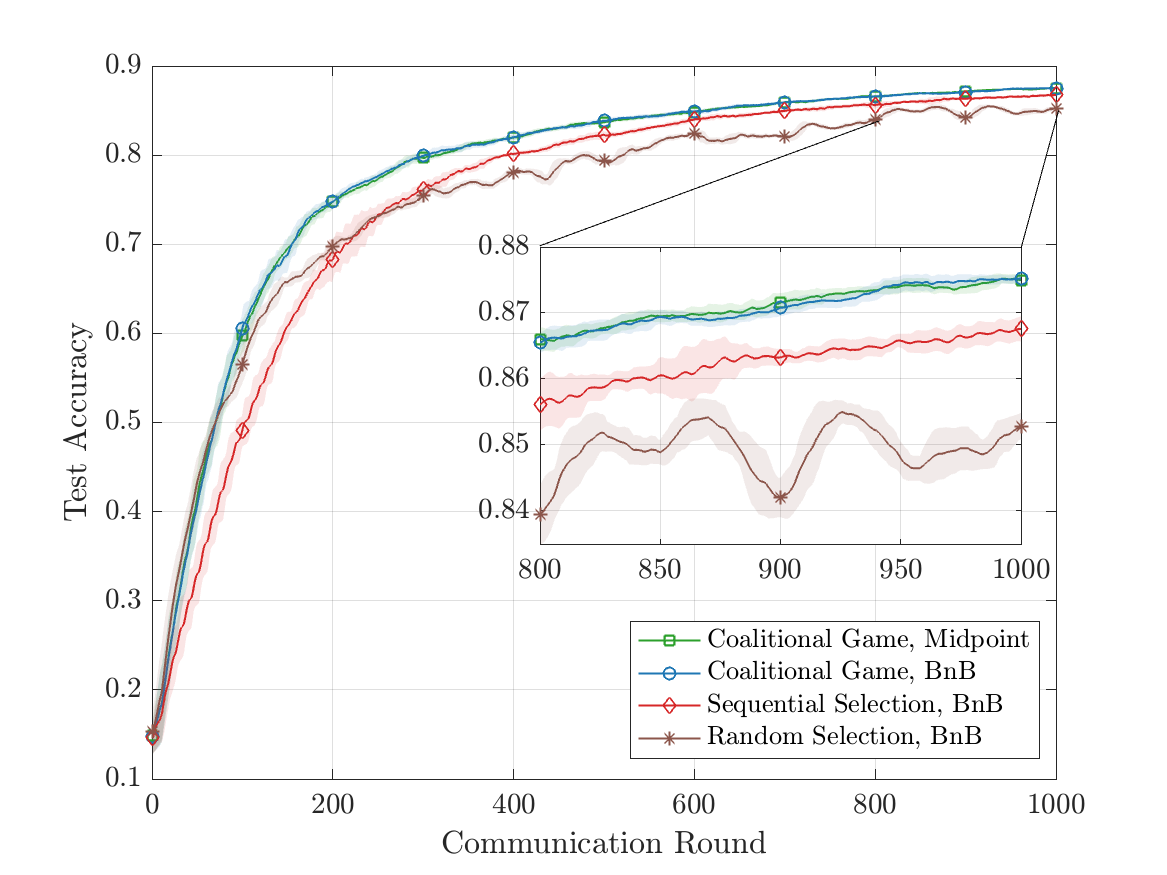}}}\vspace{-2mm}
\subfigure[Sum AoI]{\centering{\includegraphics[width=82mm]{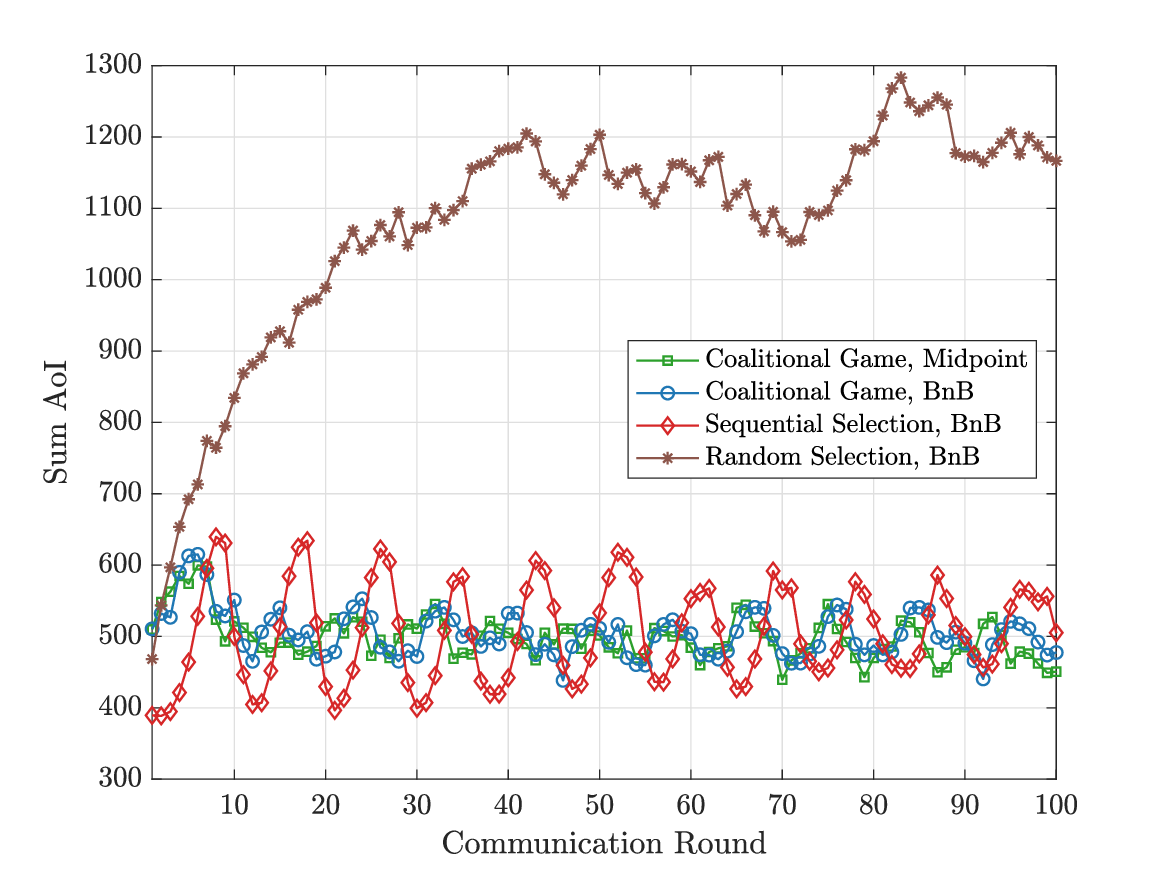}}}}\vspace{-2mm}
\caption{Learning performance on MNIST under different device selection schemes with random device locations, where $v_\mathrm{pin}=0.5$~m/s, $C_n=1$~GHz, $D_\mathrm{lm}=50$~Mbits, $P_n=20$~dBm, $T_\mathrm{max}=7$~s, and $A_0=20$.}\vspace{-4mm}
\label{result2}
\end{figure}

\begin{figure}[!t]
\centering{
\subfigure[Test Accuracy]{\centering{\includegraphics[width=82mm]{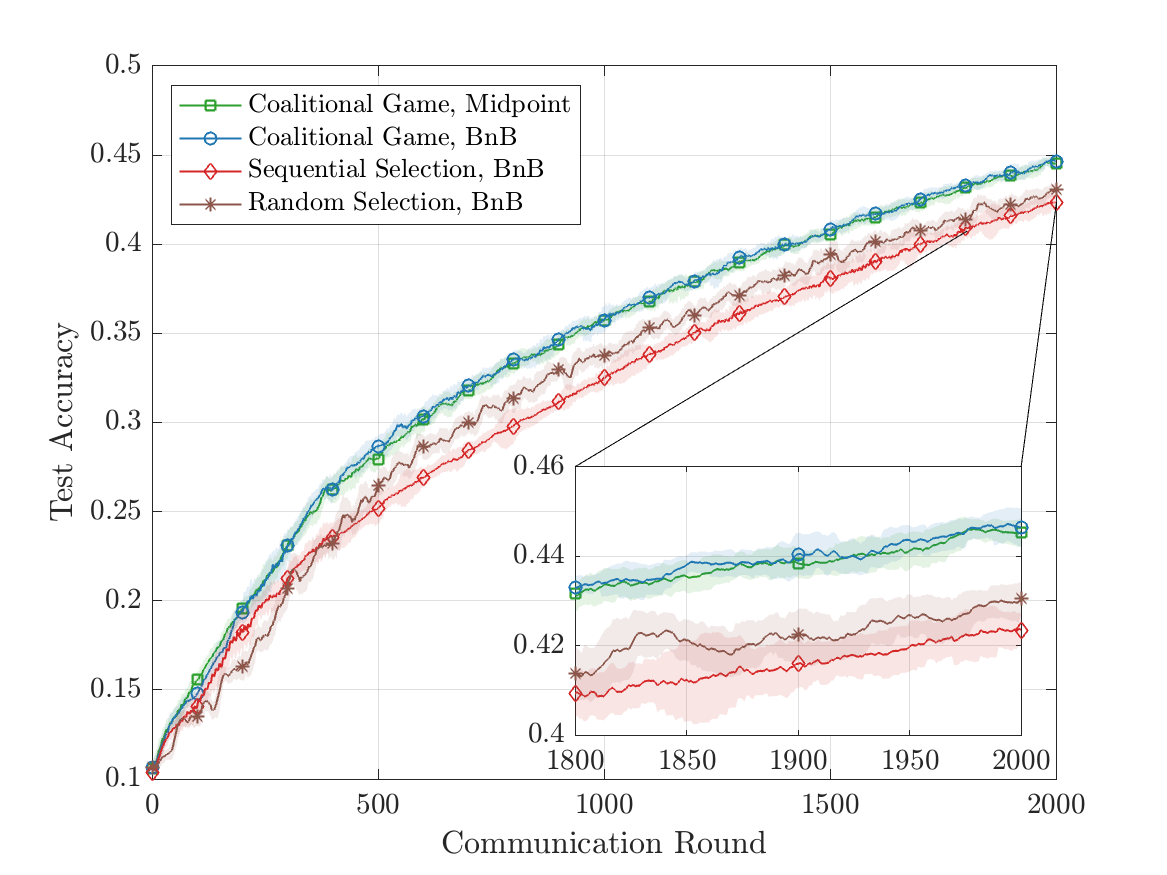}}}\vspace{-2mm}
\subfigure[Sum AoI]{\centering{\includegraphics[width=82mm]{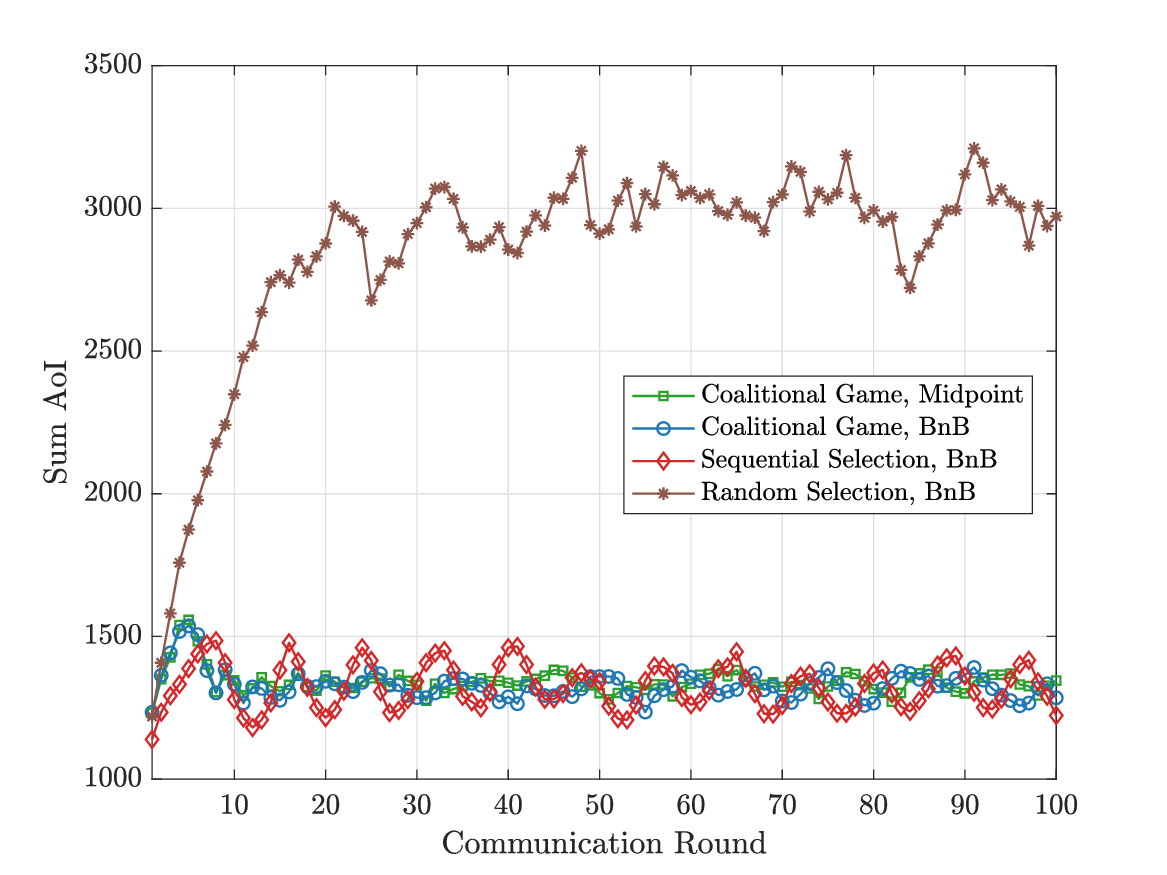}}}}\vspace{-2mm}
\caption{Learning performance on CIFAR-10 under different device selection schemes with random device locations, where $v_\mathrm{pin}=0.5$~m/s, $C_n=2$~GHz, $D_\mathrm{lm}=200$~Mbits, $P_n=20$~dBm, $T_\mathrm{max}=20$~s, and $A_0=50$.}\vspace{-4mm}
\label{result3}
\end{figure}

Figs.~\ref{result2} and \ref{result3} compare the learning performance and sum AoI of different device selection schemes on MNIST and CIFAR-10, respectively. To evaluate the performance under stochastic device deployments, the device locations are randomly regenerated in each communication round. It can be observed from Figs.~\ref{result2}(a) and \ref{result3}(a) that the proposed AoI based device selection scheme achieves the fastest convergence and highest test accuracy on both datasets, which demonstrates the effectiveness of jointly considering device freshness, round latency, and antenna placement. In contrast, the sequential selection and random selection schemes suffer from unbalanced device participation, since they cannot explicitly optimize the freshness-latency tradeoff. The midpoint based scheme is motivated by Proposition~\ref{feasibleregion}, which establishes that the optimal antenna position lies within the interval spanned by the selected devices. Under MNIST, the relatively limited feasible antenna movement range results in a narrow selected device interval, within which the midpoint solution is close to the BnB solution. For CIFAR-10, the longer local training time enlarges the feasible movement range and allows the coalitional game based scheme to select devices spanning a wider interval, making the performance advantage of the BnB based placement more pronounced. Furthermore, Figs.~\ref{result2}(b) and \ref{result3}(b) show that the proposed scheme achieves a lower sum AoI than both baseline schemes, confirming that maintaining fresher device participation is beneficial for learning under non-IID data distributions. Comparing Figs.~\ref{result2}(b) and \ref{result3}(b), the sum AoI fluctuation on CIFAR-10 is less significant, since the larger local data size also allows the pinching antenna to move over a longer distance in each round, leading to more stable uploading latency and AoI evolution.
\subsection{System Performance}
For the system performance evaluation, the device data sizes follow the same unequal distribution as that used for MNIST in the learning performance evaluation. Specifically, $9000$ training samples are distributed among $N=20$ devices, where the local data size increases from $260$ to $640$ with a step size of $20$. To obtain statistically reliable results, Monte Carlo simulations are conducted with $10^5$ independent random device deployments for each data point, where the initial antenna position in each round is set as the optimized antenna position obtained in the previous round.

\begin{figure}[!t]
\centering{\includegraphics[width=92mm]{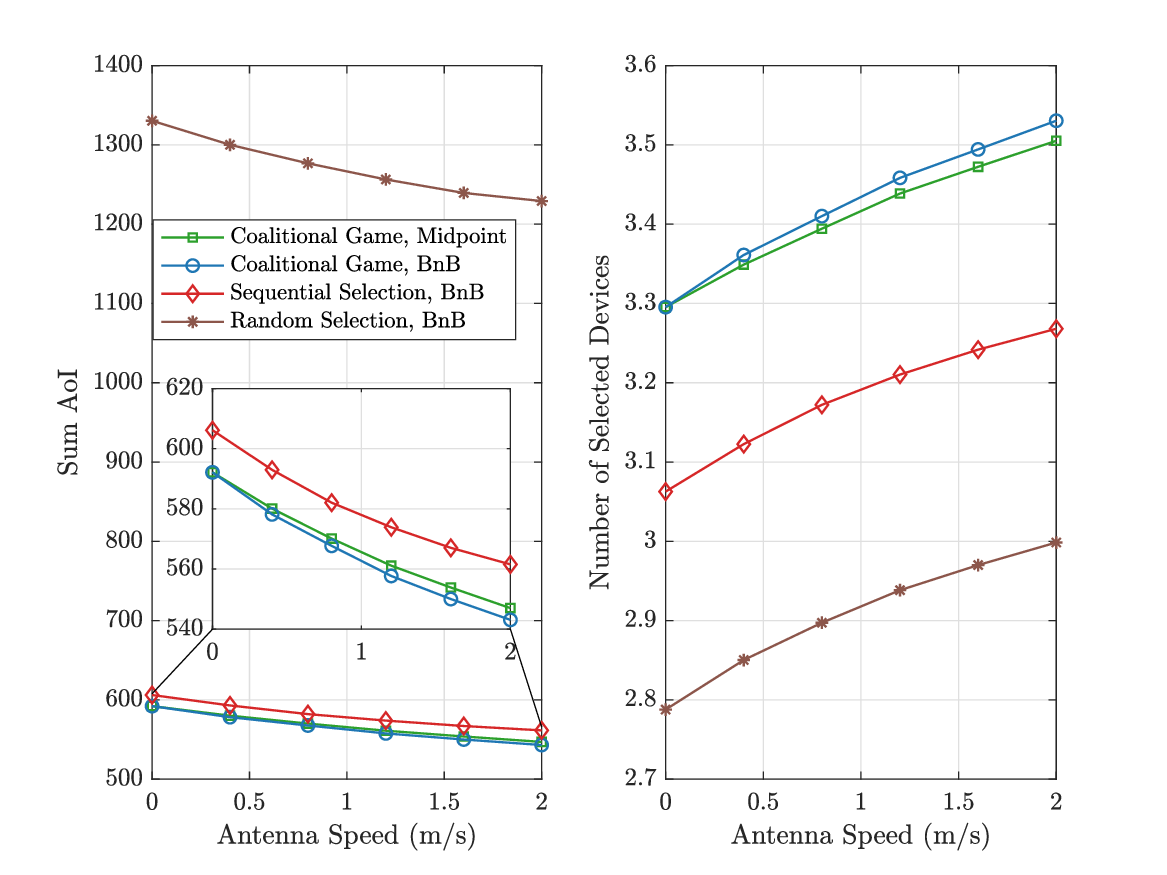}}
\caption{Impact of antenna moving speed on system performance under different schemes, where $C_n=1$~GHz, $D_\mathrm{lm}=100$~Mbits, $P_n=10$~dBm, $T_\mathrm{max}=12$~s, and $A_0=20$.}\vspace{-4mm}
\label{result4}
\end{figure}

\fref{result4} illustrates the impact of the antenna moving speed on different schemes. As $v_\mathrm{pin}$ increases, the feasible antenna placement region in each round is enlarged, which improves the uplink data rates and reduces the uploading latency. As a result, more devices can be selected under the fixed deadline $T_\mathrm{max}$, leading to a lower sum AoI. The proposed coalition game with BnB based antenna placement consistently achieves the lowest sum AoI and selects the largest number of devices, demonstrating the effectiveness of the joint AoI based device selection and antenna placement design. Moreover, the performance gap between the BnB based and midpoint based placements becomes more pronounced as $v_\mathrm{pin}$ increases. This is because a larger $v_\mathrm{pin}$ enlarges both the feasible antenna movement region and the horizontal span of the selected devices, allowing BnB to better exploit the increased placement flexibility while causing the midpoint solution to deviate further from the global optimum.

\begin{figure}[!t]
\centering{\includegraphics[width=92mm]{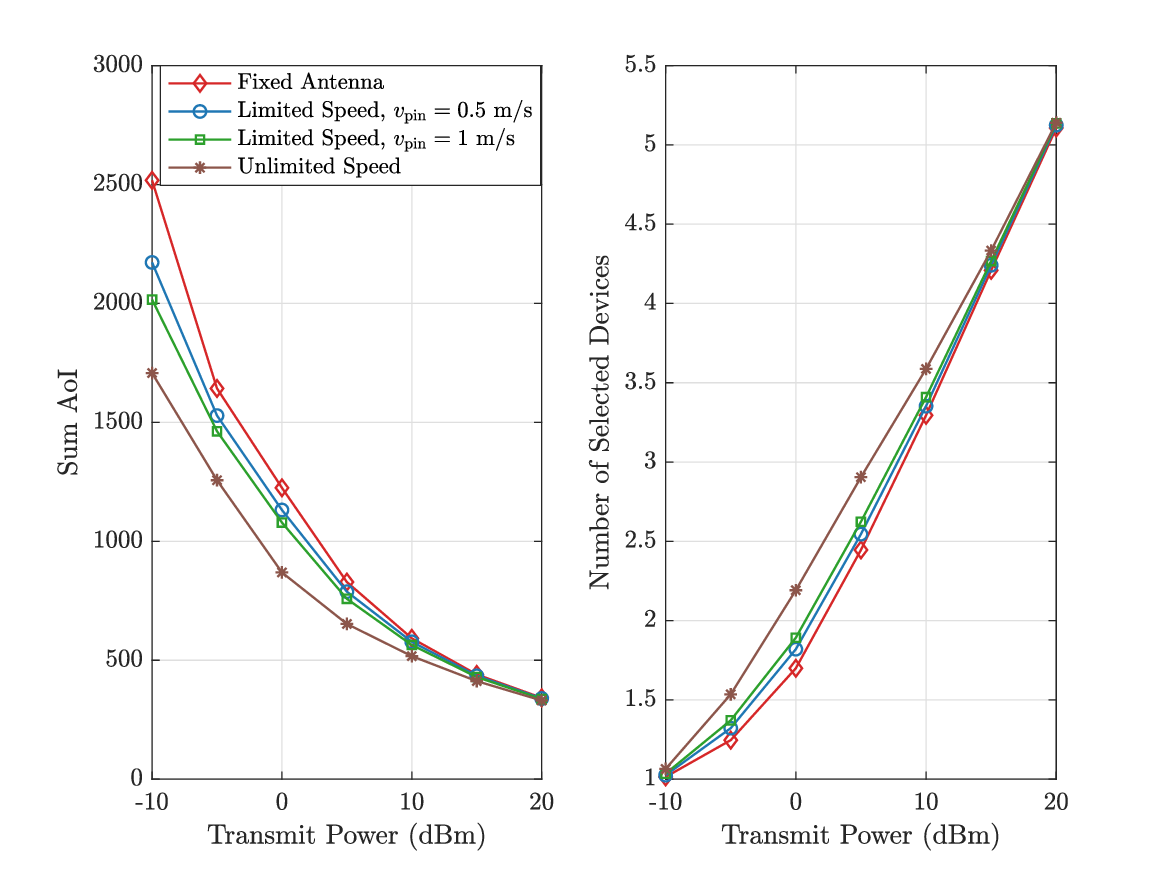}}
\caption{Impact of transmit power on system performance under different antenna moving speeds, where $C_n=1$~GHz, $D_\mathrm{lm}=100$~Mbits, $T_\mathrm{max}=12$~s, and $A_0=20$.}\vspace{-4mm}
\label{result5}
\end{figure}

\fref{result5} shows the impact of the transmit power $P_n$ on the system performance under different antenna moving speeds. As $P_n$ increases, the sum AoI decreases significantly, since a higher transmit power improves the uplink rate and reduces the model uploading time. As a result, more devices can be selected under the fixed deadline $T_\mathrm{max}$, which is consistent with the increasing number of selected devices. Overall, a larger antenna moving speed leads to a lower sum AoI and a larger selected device set, and the gain is more pronounced in the low-power region where the uplink transmission is more sensitive to channel improvement. For the number of selected devices, the curves almost overlap when $P_n$ is very low, because the uplink rate is severely limited and the mobility induced channel gain is insufficient to admit additional devices. In the medium-power region, e.g., from $-5$ to $10$ dBm, the system becomes more sensitive to antenna placement, and thus different antenna moving speeds lead to clear differences in both device participation and sum AoI. When $P_n$ is sufficiently high, the curves gradually converge again, indicating that uplink transmission is no longer the dominant bottleneck and the marginal benefit of antenna mobility becomes limited.

\begin{figure}[!t]
\centering{\includegraphics[width=92mm]{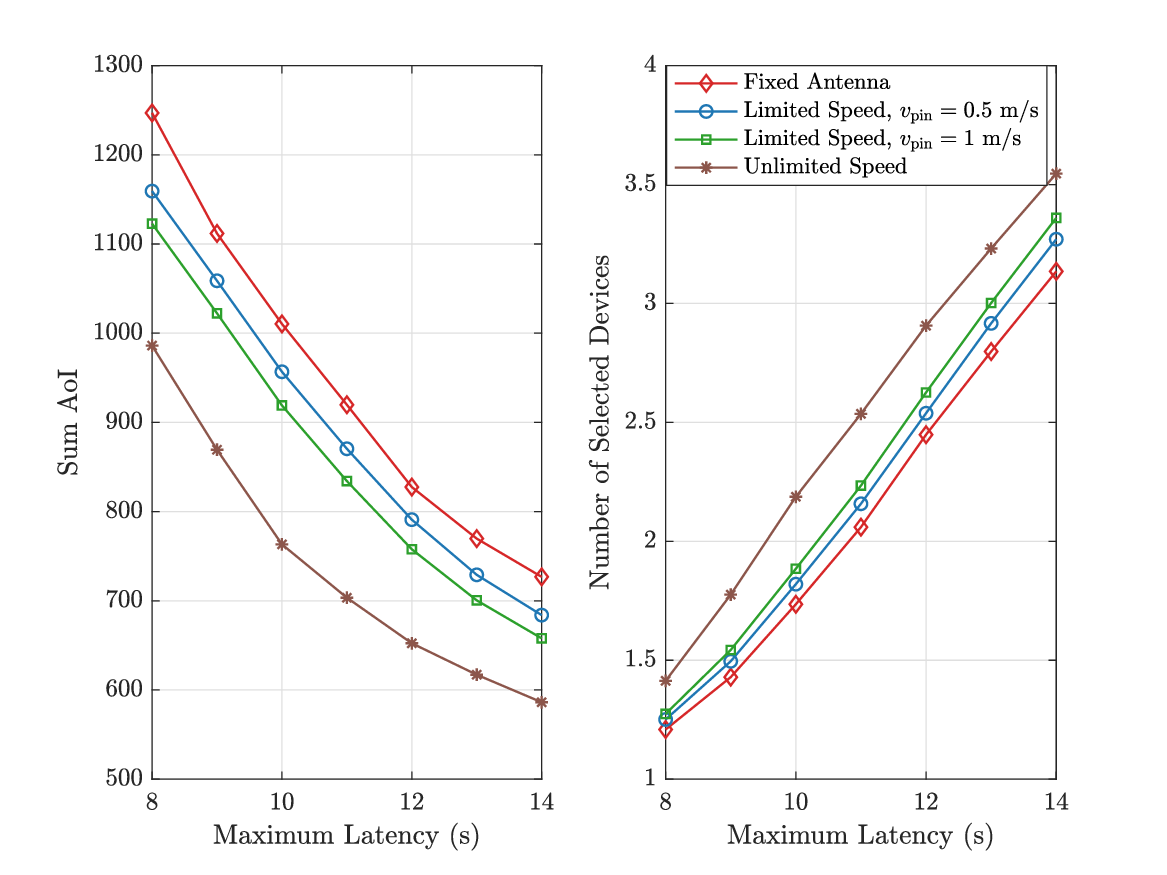}}
\caption{Impact of latency deadline on system performance under different antenna moving speeds, where $C_n=1$~GHz, $D_\mathrm{lm}=100$~Mbits, $P_n=5$~dBm, and $A_0=20$.}\vspace{-4mm}
\label{result6}
\end{figure}

\fref{result6} demonstrates the impact of the maximum latency deadline $T_\mathrm{max}$ on the system performance under different antenna moving speeds. As $T_\mathrm{max}$ increases, the sum AoI decreases, while the number of selected devices increases. This is because a more relaxed latency deadline allows more devices to complete local training and model uploading within each communication round, thereby improving device participation freshness. Moreover, a higher antenna moving speed consistently leads to a lower sum AoI and a larger selected device set. This confirms that antenna mobility can effectively improve the uplink channels and reduce uploading latency, enabling more devices to be selected under the same deadline. In particular, the unlimited speed case achieves the best performance, whereas the fixed antenna case performs the worst, indicating that the proposed AoI based device selection can achieve better system performance when supported by more flexible antenna placement.

\begin{figure}[!t]
\centering{\includegraphics[width=92mm]{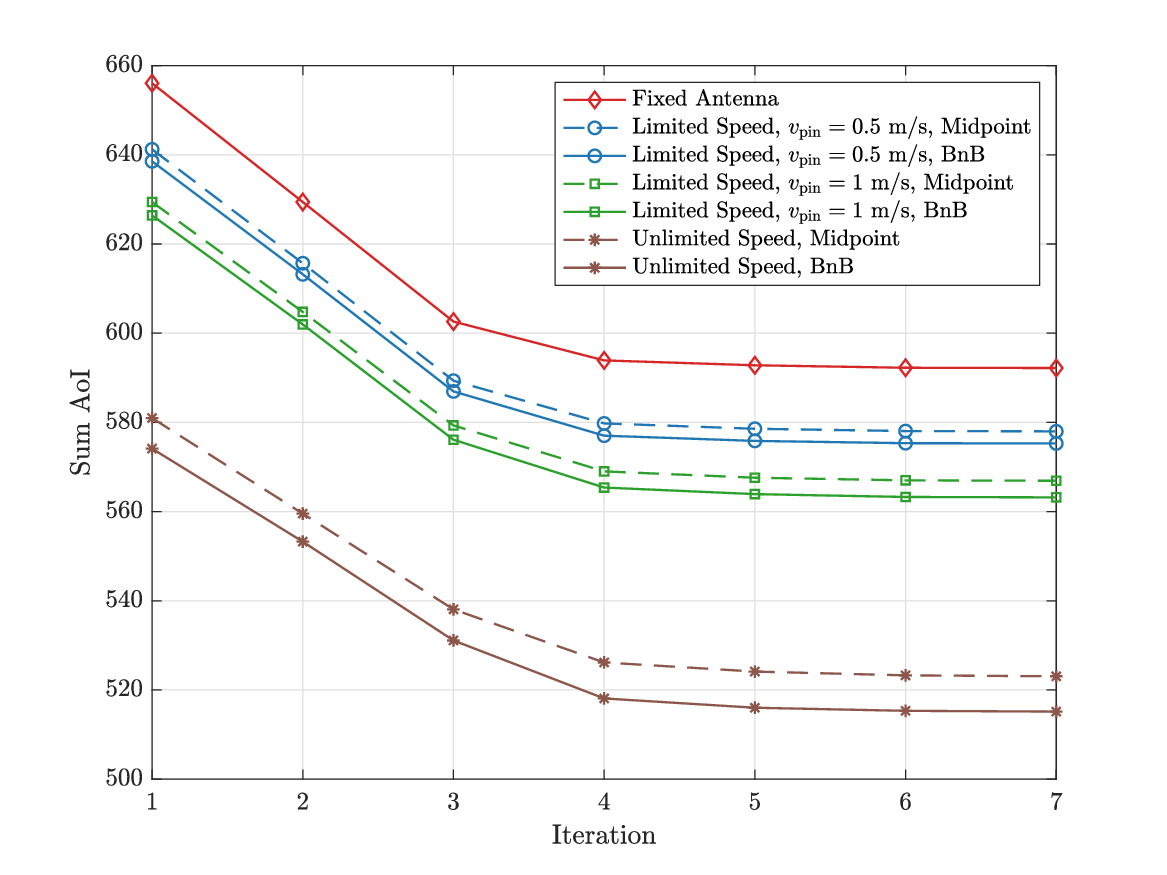}}
\caption{Convergence performance of the proposed algorithms under different antenna moving speeds, where $C_n=1$~GHz, $D_\mathrm{lm}=100$~Mbits, $P_n=10$~dBm, $T_\mathrm{max}=12$~s, and $A_0=20$.}\vspace{-4mm}
\label{result7}
\end{figure}

The convergence behavior of the proposed coalitional game based algorithm is illustrated in \fref{result7}. Since the coalition utility equals the negative overall AoI for feasible coalitions, each accepted join or leave operation strictly increases the coalition utility and therefore monotonically decreases the sum AoI. It can be observed that the sum AoI converges within a few iterations under all considered settings, demonstrating the fast convergence of the proposed algorithm. Moreover, a higher antenna moving speed enables device selection over a wider horizontal span, which leads to a lower sum AoI. Accordingly, the midpoint based placement deviates further from the globally optimal BnB solution, leading to a larger performance gap at a higher speed.
\section{Conclusions}
In this paper, an AoI aware federated learning framework with a mobility constrained pinching antenna was investigated. By utilizing the local training period for antenna repositioning, the coupling among device selection, antenna mobility, uploading latency, and AoI evolution was characterized, and an overall AoI minimization problem was formulated under a round latency deadline. To solve this problem, a coalitional game based device selection algorithm and a BnB based antenna placement algorithm were developed. Simulation results demonstrated that the proposed scheme improves learning convergence, reduces the sum AoI, and supports more device participation compared with baseline schemes, confirming that pinching antennas can provide an effective spatial reconfiguration capability for enhancing wireless federated learning. As a future extension, incorporating long-term performance over multiple communication rounds into the current optimization formulation represents an interesting research direction.
\section*{Appendix~A: Proof of Proposition~\ref{rategap}}
The non-saturated regime $v_\mathrm{pin}<v_{\mathrm{sat},n}$ is first considered. In this regime, the residual horizontal distance between the pinching antenna and the projection of device $n$ is $r_n(v_\mathrm{pin})$, while the corresponding horizontal distance under the fixed antenna benchmark is $|x_n|$. Therefore, the rate gain from the fixed antenna to the speed limited pinching antenna is
\begin{align}
\Delta R_{n,0\to v}&=B\log_2\!\left(\!1\!+\!\frac{a_n}{r_n^2(v_\mathrm{pin})\!+\!b_n}\!\right)\!-\!B\log_2\!\left(\!1\!+\!\frac{a_n}{x_n^2\!+\!b_n}\!\right)\nonumber\\
&=B\log_2\!\left(\!1\!+\!\frac{a_n\left(x_n^2-r_n^2(v_\mathrm{pin})\right)}{\left(r_n^2(v_\mathrm{pin})\!+\!b_n\right)\left(x_n^2\!+\!b_n\!+\!a_n\right)}\!\right).
\end{align}
Since $0\le r_n(v_\mathrm{pin})\le |x_n|$, it follows that $x_n^2-r_n^2(v_\mathrm{pin})\ge 0$. By applying $\ln(1+z)\ge \frac{z}{1+z}$ for $z\ge 0$, the following lower bound is obtained:
\begin{equation}
\Delta R_{n,0\to v}\ge\frac{B}{\ln 2}\frac{a_n\left(x_n^2-r_n^2(v_\mathrm{pin})\right)}{\left(x_n^2+b_n\right)\left(r_n^2(v_\mathrm{pin})+b_n+a_n\right)}.
\end{equation}
To characterize its monotonicity with respect to $v_\mathrm{pin}$, define
\begin{equation}
f(r)=\frac{x_n^2-r^2}{r^2+b_n+a_n},
\end{equation}
where $0\le r\le |x_n|$. The derivative of $f(r)$ is
\begin{equation}
\frac{d f(r)}{dr}=-\frac{2r\left(x_n^2+b_n+a_n\right)}{\left(r^2+b_n+a_n\right)^2}\le 0.
\end{equation}
Thus, $f(r)$ is non-increasing in $r$. Since $r_n(v_\mathrm{pin})$ is non-increasing in $v_\mathrm{pin}$, the lower bound of $\Delta R_{n,0\to v}$ is non-decreasing in $v_\mathrm{pin}$.

Next, the rate gap from the speed limited pinching antenna to the unlimited speed benchmark is given by
\begin{align}
\Delta R_{n,v\to\infty}&=B\log_2\!\left(\!1\!+\!\frac{a_n}{b_n}\!\right)\!-\!B\log_2\!\left(\!1\!+\!\frac{a_n}{r_n^2(v_\mathrm{pin})\!+\!b_n}\!\right) \nonumber\\
&=B\log_2\!\left(\!1\!+\!\frac{a_n r_n^2(v_\mathrm{pin})}{b_n\left(r_n^2(v_\mathrm{pin})\!+\!b_n\!+\!a_n\right)}\!\right).
\end{align}
By applying $\ln(1+z)\le z$ for $z\ge 0$, one has
\begin{equation}
\Delta R_{n,v\to\infty}\le\frac{B}{\ln 2}\frac{a_n r_n^2(v_\mathrm{pin})}{b_n\left(r_n^2(v_\mathrm{pin})+b_n+a_n\right)}.
\end{equation}
Define the speed dependent term in the above upper bound as
\begin{equation}
g(r)=\frac{r^2}{r^2+b_n+a_n}.
\end{equation}
For $r\ge 0$, its derivative is
\begin{equation}
\frac{d g(r)}{dr}=\frac{2r(b_n+a_n)}{\left(r^2+b_n+a_n\right)^2}\ge 0.
\end{equation}
Thus, $g(r)$ is non-decreasing in $r$. Since $r_n(v_\mathrm{pin})$ is non-increasing in $v_\mathrm{pin}$, the upper bound of $\Delta R_{n,v\to\infty}$ is non-increasing in $v_\mathrm{pin}$.

Finally, when $v_\mathrm{pin}\ge v_{\mathrm{sat},n}$, the pinching antenna reaches the projection of device $n$ during local training, which leads to $r_n(v_\mathrm{pin})=0$. Hence, the speed limited pinching antenna achieves the same rate as the unlimited speed benchmark, and $\Delta R_{n,v\to\infty}=0$. The proof is completed.\QEDA
\section*{Appendix~B: Proof of Proposition~\ref{feasibleregion}}
Based on \eqr{noruptime}, since $a_n>0$, $f_n(x_\mathrm{pin}^{(t)})$ is monotonically increasing with respect to $(x_\mathrm{pin}^{(t)}-x_n)^2$. Consider any feasible antenna position satisfying $x_\mathrm{pin}^{(t)}<x_\mathrm{min}^{(t)}$. For any selected device $n\in\mathcal{S}_t$, it follows that
\begin{equation}
|x_\mathrm{pin}^{(t)}-x_n|>|x_\mathrm{min}^{(t)}-x_n|,
\end{equation}
which implies
\begin{equation}
f_n(x_\mathrm{pin}^{(t)})>f_n(x_\mathrm{min}^{(t)}).
\end{equation}
Therefore, the following inequality can be obtained:
\begin{equation}
\sum_{n\in\mathcal{S}_t} f_n(x_\mathrm{pin}^{(t)})>\sum_{n\in\mathcal{S}_t} f_n(x_\mathrm{min}^{(t)}),
\end{equation}
which indicates that no feasible antenna position satisfying $x_\mathrm{pin}^{(t)}<x_\mathrm{min}^{(t)}$ can be globally optimal.

Similarly, for any feasible antenna position satisfying $x_\mathrm{pin}^{(t)}>x_\mathrm{max}^{(t)}$, the following inequality holds for all $n\in\mathcal{S}_t$:
\begin{equation}
f_n(x_\mathrm{pin}^{(t)})>f_n(x_\mathrm{max}^{(t)}),
\end{equation}
and hence
\begin{equation}
\sum_{n\in\mathcal{S}_t} f_n(x_\mathrm{pin}^{(t)})>\sum_{n\in\mathcal{S}_t} f_n(x_\mathrm{max}^{(t)}),
\end{equation}
which leads to the conclusion that no feasible position satisfying $x_\mathrm{pin}^{(t)}>x_\mathrm{max}^{(t)}$ can be globally optimal.

Therefore, the global optimum must lie in the intersection between the mobility constrained feasible interval $[l^{(t)},u^{(t)}]$ and the span of the selected devices $[x_\mathrm{min}^{(t)},x_\mathrm{max}^{(t)}]$. When $\underline{x}^{(t)}\le\overline{x}^{(t)}$, this intersection is exactly $[\underline{x}^{(t)},\overline{x}^{(t)}]$, which completes the proof.\QEDA
\bibliographystyle{IEEEtran}
\bibliography{KaidisBib}
\end{document}